%% file: main.tex
\PassOptionsToPackage{svgnames}{xcolor}
\documentclass[sigconf, nonacm]{acmart}

\usepackage[utf8]{inputenc}
\usepackage[T1]{fontenc}
\usepackage{tlatex}
\usepackage{color}
\usepackage{algorithm}
\usepackage{algpseudocode}
\usepackage{svg}
\usepackage{subcaption}
\usepackage{verbatim}
\usepackage{xcolor}
\definecolor{boxshade}{gray}{.85}
\setboolean{shading}{true}

\newcommand\vldbdoi{XX.XX/XXX.XX}
\newcommand\vldbpages{XXX-XXX}
\newcommand\vldbvolume{20}
\newcommand\vldbissue{1}
\newcommand\vldbyear{2026}
\newcommand\vldbauthors{\authors}
\newcommand\vldbtitle{\shorttitle} 
\newcommand\vldbavailabilityurl{https://github.com/oceanbase/oceanbase/tree/develop/src/storage/tx}
\newcommand\vldbpagestyle{plain} 

\begin{document}
\title{A Tree-Structured Two-Phase Commit Framework for OceanBase: Optimizing Scalability and Consistency}

\settopmatter{authorsperrow=4}

\author{Quanqing Xu}
\affiliation{%
  \institution{OceanBase, Ant Group}
}
\authornote{Quanqing Xu and Chen Qian contributed equally to this work.}

\author{Chen Qian}
\affiliation{%
  \institution{OceanBase, Ant Group}
}
\authornotemark[1]

\author{Chuanhui Yang}
\affiliation{%
  \institution{OceanBase, Ant Group}
}
\authornote{Chuanhui Yang and Fanyu Kong are corresponding authors.}

\author{Fanyu Kong}
\affiliation{%
  \institution{OceanBase, Ant Group}
}
\authornotemark[2]

\author{Guixiang Liu}
\affiliation{%
  \institution{OceanBase, Ant Group}
}

\author{Fusheng Han}
\affiliation{%
  \institution{OceanBase, Ant Group}
}

\author{Zixiang Zhai}
\affiliation{%
  \institution{OceanBase, Ant Group}
}







\begin{abstract}
Modern distributed databases face challenges in achieving transactional consistency across distributed partitions. Traditional two-phase commit (2PC) protocols incur high coordination overhead and latency, and require complex recovery for dynamic partition transfers. This paper introduces a novel tree-shaped 2PC framework for OceanBase that leverages single-machine log streams to address these challenges through three innovations. First, we propose log streams as atomic participants, replacing partition-level coordination. By treating each log stream as the commit unit, a transaction spanning $N$ co-located partitions interacts with one participant, reducing coordination overhead by orders of magnitude (e.g., 99\% reduction for $N=100$). Second, we design a tree-shaped 2PC protocol with coordinator-rooted DAG topology that dynamically handles partition transfers by recursively constructing commit trees. When a partition migrates during a transaction, the protocol embeds migration contexts as leaf nodes, eliminating explicit participant list updates, resolving circular dependencies, and ensuring linearizable commits under topology changes. Third, we introduce prepare\_unknown and trans\_unknown states to prevent consistency violations when participants lose context. These states signal uncertainty during retries, avoiding erroneous aborts from ``lying'' participants while isolating users from ambiguity. Experimental evaluation demonstrates performance approaching that of single-machine transactions, with reduced latency and bandwidth consumption, validating the framework's effectiveness for modern distributed databases.
\end{abstract}

\maketitle

\pagestyle{\vldbpagestyle}
\begingroup\small\noindent\raggedright\textbf{PVLDB Reference Format:}\\
\vldbauthors. \vldbtitle. PVLDB, \vldbvolume(\vldbissue): \vldbpages, \vldbyear.\\
\href{https://doi.org/\vldbdoi}{doi:\vldbdoi}
\endgroup
\begingroup
\renewcommand\thefootnote{}\footnote{\noindent
This work is licensed under the Creative Commons BY-NC-ND 4.0 International License. Visit \url{https://creativecommons.org/licenses/by-nc-nd/4.0/} to view a copy of this license. For any use beyond those covered by this license, obtain permission by emailing \href{mailto:info@vldb.org}{info@vldb.org}. Copyright is held by the owner/author(s). Publication rights licensed to the VLDB Endowment. \\
\raggedright Proceedings of the VLDB Endowment, Vol. \vldbvolume, No. \vldbissue\ %
ISSN 2150-8097. \\
\href{https://doi.org/\vldbdoi}{doi:\vldbdoi} \\
}\addtocounter{footnote}{-1}\endgroup

\ifdefempty{\vldbavailabilityurl}{}{
\vspace{.3cm}
\begingroup\small\noindent\raggedright\textbf{PVLDB Artifact Availability:}\\
The source code, data, and/or other artifacts have been made available at \url{\vldbavailabilityurl}.
\endgroup
}

\section{Introduction}
\label{sec:intro}
In the era of large-scale distributed databases, ensuring transactional consistency across geographically dispersed data partitions remains a fundamental challenge. As modern applications increasingly rely on distributed database systems to handle massive workloads, the need for efficient transaction coordination mechanisms becomes paramount. Traditional 2PC protocols~\cite{lampson1993new, atif2009analysis, desai1996performance, samaras1995two} have served as the cornerstone for maintaining ACID properties in distributed environments, providing a reliable framework for coordinating transaction commits across multiple participants. However, as database systems scale to support thousands of partitions and millions of concurrent transactions, these foundational protocols face significant challenges in meeting the performance and scalability requirements of modern cloud-native applications.

Traditional 2PC protocols face two critical limitations. First, \textit{scalability bottlenecks} arise from partition-centric designs that treat each partition as an independent participant: a transaction spanning 100 partitions incurs 100 prepare/commit phases even when all reside on one machine, causing quadratic growth in coordination messages and log synchronization and severely limiting throughput under high concurrency. Second, \textit{complicated dynamic topologies} emerge during partition transfers, essential for load balancing in elastic clouds. Existing approaches either require costly re-validation of participant lists post-transfer~\cite{ezhilchelvan2018non} or introduce circular dependencies between partitions and logs~\cite{kolltveit2007circular, alkhatib2002transaction}, undermining consistency when partitions migrate during active transactions.

Recent advances in distributed transaction management have made significant strides, yet fundamental limitations persist. Systems like Spanner~\cite{corbett2013spanner} and CockroachDB~\cite{taft2020cockroachdb} still rely on per-partition coordination, inheriting the scalability bottlenecks of traditional approaches. While log-structured systems demonstrate the potential of log-centric architectures, they do not integrate log streams as atomic units for transaction coordination, leaving the scalability and dynamic topology challenges unresolved.

In OceanBase's single-machine log stream architecture, co-located partition leaders within a unit are placed on the same machine, enabling all these partitions to write to a shared log stream. This architectural choice provides a unique opportunity to optimize 2PC: by using the log stream as a participant instead of individual partitions, we can dramatically reduce coordination overhead while naturally handling partition transfers through the log stream's inherent structure. This paper proposes a novel log-stream-based 2PC framework that addresses the scalability and dynamic topology challenges through three key innovations. Our approach fundamentally shifts the atomic unit from partitions to log streams, introduces a dynamic tree-structured protocol that adapts to partition migrations, and provides a principled mechanism for handling context loss scenarios. 

Therefore, the contributions of this paper are as follows:

\begin{itemize}
    \item \textbf{Log Stream as Participant}: We abstract log streams as the new atomic unit instead of individual partitions, reducing participant count by orders of magnitude. A transaction with 100 co-located partitions now interacts with one log stream participant, cutting coordination overhead by 99\% and achieving single-machine transaction performance. 
    
    \item \textbf{Dynamic Tree-Structured 2PC}: Unlike traditional hierarchical 2PC \cite{transactionbook} with static structures, we introduce a \textit{dynamic} tree-structured protocol where the participant list forms a DAG that evolves as partitions migrate. The structure recursively tracks log stream dependencies, automatically constructing trees on-the-fly and eliminating explicit participant list updates, resolving the ``circular transaction'' problem \cite{alkhatib2002transaction}.
    
    \item \textbf{Unknown-State Mechanism for Consistency Preservation}: When participants lose transaction context, they may ``lie'' about prior decisions. We introduce \textit{prepare\_unknown} and \textit{trans\_unknown} states: participants signal uncertainty, and recreated coordinators return \textit{trans\_unknown} to users, preventing forced commits while shielding users from ambiguous states (detailed in Section~\ref{sec:unknown}).
\end{itemize}

Our contributions advance both theoretical understanding and practical implementation: the log stream abstraction challenges the assumption that partitions must be atomic units; the dynamic tree-structured protocol handles changing topologies during execution; and the unknown-state mechanism ensures consistency under context loss. This framework meets two design goals: (1) \textbf{high performance for static transactions}—achieving single-machine transaction latency (1.2ms) through log stream aggregation; and (2) \textbf{simplicity for dynamic transactions}—handling partition transfers with minimal complexity through automatic tree construction. Rigorous evaluation on OceanBase demonstrates linear scaling in commit throughput under heavy transfer workloads, positioning our solution as a critical advancement for next-generation distributed databases in cloud environments.

The paper is organized as follows. \S\ref{sec:background} provides the background on the OceanBase database, ACID properties, and log stream topologies during partition transfer. \S\ref{sec:tree_2pc} describes the tree-shaped 2PC protocol, including optimizations for transaction latency and the state machine design. In \S\ref{sec:transfer}, we discuss the transfer process and how it handles concurrent partition transfers during transaction execution. \S\ref{sec:unknown} presents the unknown-state mechanism for consistency preservation. \S\ref{sec:discussion} discusses optimizations and the TLA+ proof, and \S\ref{sec:experimental_evaluation} presents the experimental evaluation to demonstrate the high performance and scalability of our framework. We review the related work in \S\ref{sec:related_works} and present the conclusions in \S\ref{sec:Conclusions}.

\section{Background}
\label{sec:background}

\subsection{OceanBase Database}
\label{sec:background:ob}

OceanBase is a high-performance distributed relational database system engineered to address the scalability, cross-regional fault tolerance, and cost-efficiency demands of modern internet-scale applications~\cite{DBLP:journals/pvldb/YangYHZYYCZSXYL22}. Employing a shared-nothing architecture, OceanBase emphasizes superior scalability and high availability. The latest iteration, OceanBase 4.0, introduces Paetica~\cite{paetica}, an innovative architecture that merges the shared-nothing and shared-everything paradigms, enabling seamless transitions between standalone and distributed modes.

A key architectural feature of OceanBase is its single-machine log stream design, where co-located partition leaders within a unit are placed on the same machine, enabling all these partitions to write to a shared log stream. This design provides the foundation for our log-stream-based 2PC optimization, as it naturally aggregates multiple partitions into a single log stream participant, reducing coordination overhead while maintaining transaction consistency.

\subsection{ACID Property and Traditional 2PC}
\label{sec:background:ACID}

Traditional 2PC is a fundamental protocol for ensuring atomicity in distributed transactions~\cite{lampson1993new}. The protocol consists of two phases: (1) \textit{prepare phase}, where the coordinator sends prepare requests to all participants, and each participant votes YES (by writing a prepare log) or NO; (2) \textit{commit/abort phase}, where the coordinator collects all votes and sends commit or abort messages to all participants based on the voting results. If all participants vote YES, the transaction commits; otherwise, it aborts. This protocol ensures that all participants either commit or abort together, maintaining transaction atomicity.

In the single-machine log stream solution, partitions remain the fundamental units for transactional read and write operations, and concurrency control still operates at the partition level. To ensure ACID characteristics~\cite{xie2014salt}, the implementation of atomicity requires that all units that have participated in concurrency control successfully pass isolation and persistence checks. Therefore, if we use log streams as participants in the 2PC protocol, the minimum set of participants must include all log streams where partitions participating in concurrency control are finally located. We refer to this as the \textbf{minimum set requirement}. Any participant list containing this minimum set is also correct. The key challenge lies in how to obtain all such log streams during the 2PC phase and complete the overall commit process.

\subsection{Log Stream Topologies}
\label{sec:background:design}

Partition transfer is a critical mechanism in OceanBase for load balancing and resource optimization. However, it introduces complexity in transaction coordination: when partitions migrate during an active transaction, the participant list must be updated to include all log streams that have hosted the transaction's partitions.

In a naive approach, the coordinator constructs the participant list based only on log streams where partitions performed read/write operations at transaction start. However, if partition transfer occurs before transaction commit, this list violates the minimum set requirement (Section~\ref{sec:background:ACID}). As illustrated in Figure~\ref{fig:1}, a transaction initially involves log streams A and B, but after partition transfer, log stream C also creates transaction context and must participate in concurrency control. The naive approach misses log stream C, leading to incorrect participant lists.

\begin{figure}
  \centering
  \includegraphics[width=\linewidth]{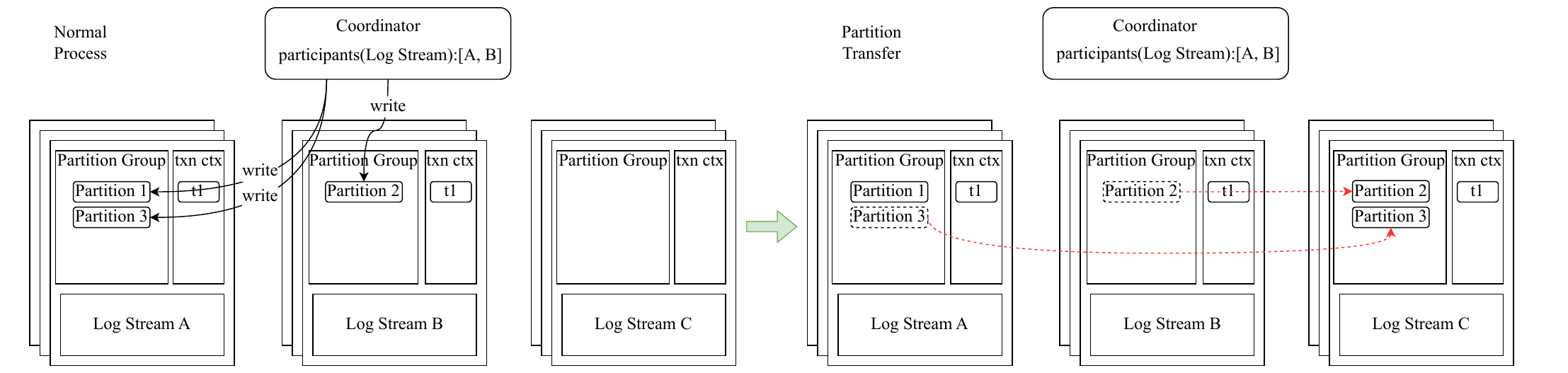}
  \caption{Participant list changes during a transaction.}
  \label{fig:1}
\end{figure}

To correctly identify all relevant log streams, we observe that transaction contexts are created in two scenarios: (1) when partitions perform read/write operations, and (2) when partitions are transferred. The log streams from read/write operations are directly available, while those from transfers can be recursively obtained through transfer records: when a partition transfers from log stream $L_s$ to $L_d$, the source log stream records the destination in its transaction context. By recursively following these transfer records, we can construct a \textit{log stream tree} that includes all log streams that have hosted the transaction's partitions. Example~\ref{example:participant_list_change} illustrates this construction, and Theorem~\ref{thm:log_stream_tree} establishes that this tree satisfies the minimum set requirement.

\begin{figure}
  \centering
  \includegraphics[width=\linewidth]{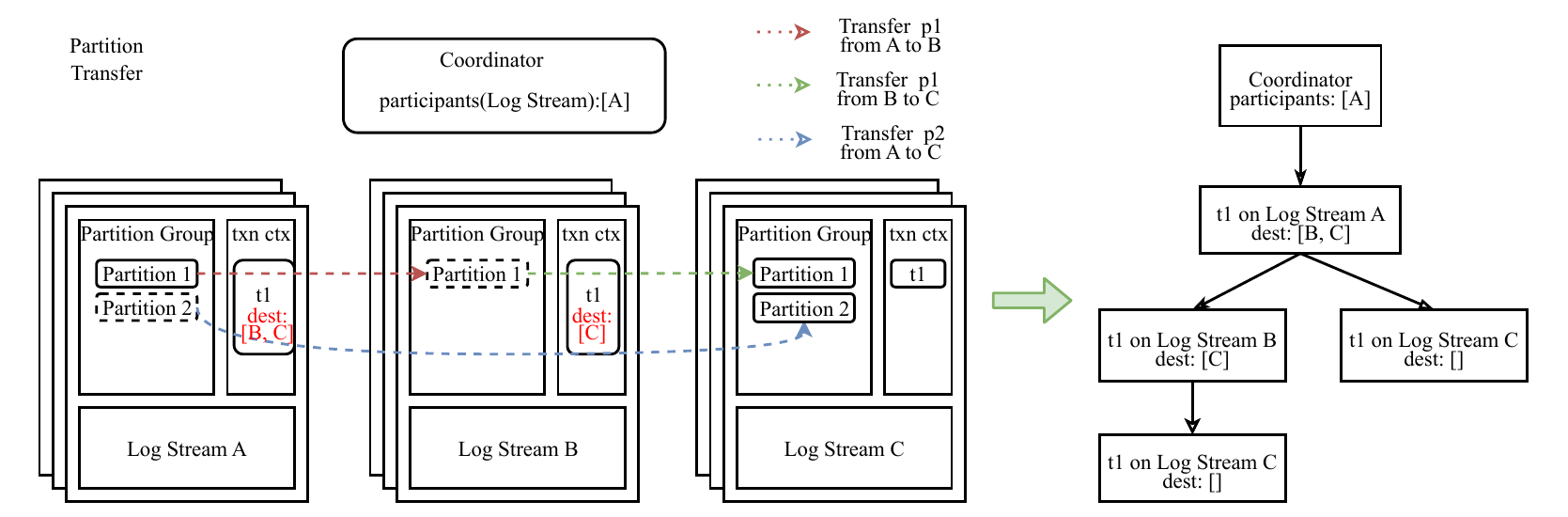}
  \caption{Log stream tree during execution.}
  \label{fig:2}
\end{figure}

\begin{example} \label{example:participant_list_change}
    As shown on the left side of Figure~\ref{fig:2}, transaction t1 involves partitions Partition1 and Partition2 during execution, both of which are on Log Stream (LS) A. Afterwards, Partition1 is first transferred from Log Stream A to Log Stream B; then from Log Stream B to Log Stream C; Partition2 is transferred from Log Stream A to Log Stream C. The success of the transfer enables transaction t1 on Log Stream A to record the destinations Log Stream B and Log Stream C generated by the transfer; similarly, transaction t1 on Log Stream B records the destination Log Stream C generated by the transfer. The log stream tree is shown on the right side, where the coordinator is the head node and the log streams where the partitions participating in concurrency control are finally located serve as the leaf nodes. 
\end{example}  

\begin{theorem} \label{thm:log_stream_tree}
    The log stream tree meets the minimum set requirement.
\end{theorem}
The proof is given in Appendix~\ref{sec:appendix:log_stream_tree}.


\section{Tree-Shaped 2PC}
\label{sec:tree_2pc}

\subsection{Tree Structure Design}
\label{sec:tree_2pc:design}

Following Theorem~\ref{thm:log_stream_tree}, the log stream tree satisfies the minimum set requirement for the participant list. Tree-shaped 2PC (a.k.a., hierarchical 2PC)~\cite{transactionbook}, which organizes coordinator and participants into a tree, naturally fits our design. Each node with nonempty children acts as both a coordinator (for its children) and a participant (for its parent); the set of nodes in its subtrees forms its participant list. Nodes without children act only as participants. 

However, traditional hierarchical 2PC still has room for improvement in transaction response time. To this end, this section introduces optimizations of the traditional 2PC protocol in OceanBase, then combines them with the tree-shaped structure to achieve high performance while maintaining correctness.


\subsection{Optimization for Transaction Latency}
\label{sec:tree_2pc:optimization}

    



Traditional 2PC consists of two stages: prepare and commit/abort. It requires the coordinator and participants to synchronously write their prepare\footnote{The prepare log for coordinator may be omitted in some implementations of 2PC protocol.} and commit/abort logs. The coordinator replies to the client after its commit log is persisted, incurring two RPC messages and three log synchronizations as transaction latency. Since the two RPC messages are mandatory, our optimization focuses on reducing I/O operations.

OceanBase 2PC optimizes this by eliminating coordinator logging: the coordinator never writes logs. However, this creates two challenges: (a) in abnormal scenarios, the coordinator cannot determine all participants' states without its own prepare log, so each participant must record all participants in its prepare log; (b) participants cannot confirm the final transaction status through the coordinator due to the lack of commit/abort logs, requiring a mechanism to ensure participants can independently determine transaction status before releasing resources.

Initially, OB 2PC addressed these challenges by adding a clear stage: all participants synchronously write a clear log after commit/abort before entering TOMBSTONE state. The coordinator responds to the client after receiving prepare OK/NO messages, reducing transaction latency to two RPC messages and one I/O operation (participants' prepare logs). However, clear logs increase resource consumption. In later versions, the clear stage is removed: the coordinator synchronously writes the commit/abort log \textit{after} replying to the client, maintaining one I/O operation during transaction response time while ensuring correctness.

\subsection{State Machine}
\label{sec:tree_2pc:state_machine}

Following the optimization in Section~\ref{sec:tree_2pc:optimization}, we present the tree-shaped 2PC protocol in OceanBase. The participant list consists of the transaction execution log stream or the target log stream of partition transfer. Transfer destinations are recorded in the source transaction context's participant list, with detailed implementation in Section~\ref{sec:transfer}. To simplify the design, we abstract all tree nodes uniformly without distinguishing between leaf and non-leaf nodes. Each node maintains the following variables:
\begin{itemize}
    \item \textit{status}: The node's own status (OK or NO), initially OK.
    \item \textit{state}: The 2PC state (RUNNING, PREPARE, COMMIT, ABORT, TOMBSTONE), initially RUNNING.
    \item \textit{participant list}: The list of child participants for the current state, initially empty.
    \item \textit{collected participant list}: The list of participants whose messages have been collected, initially empty.
\end{itemize}

\begin{figure}
  \centering
  \includegraphics[width=\linewidth]{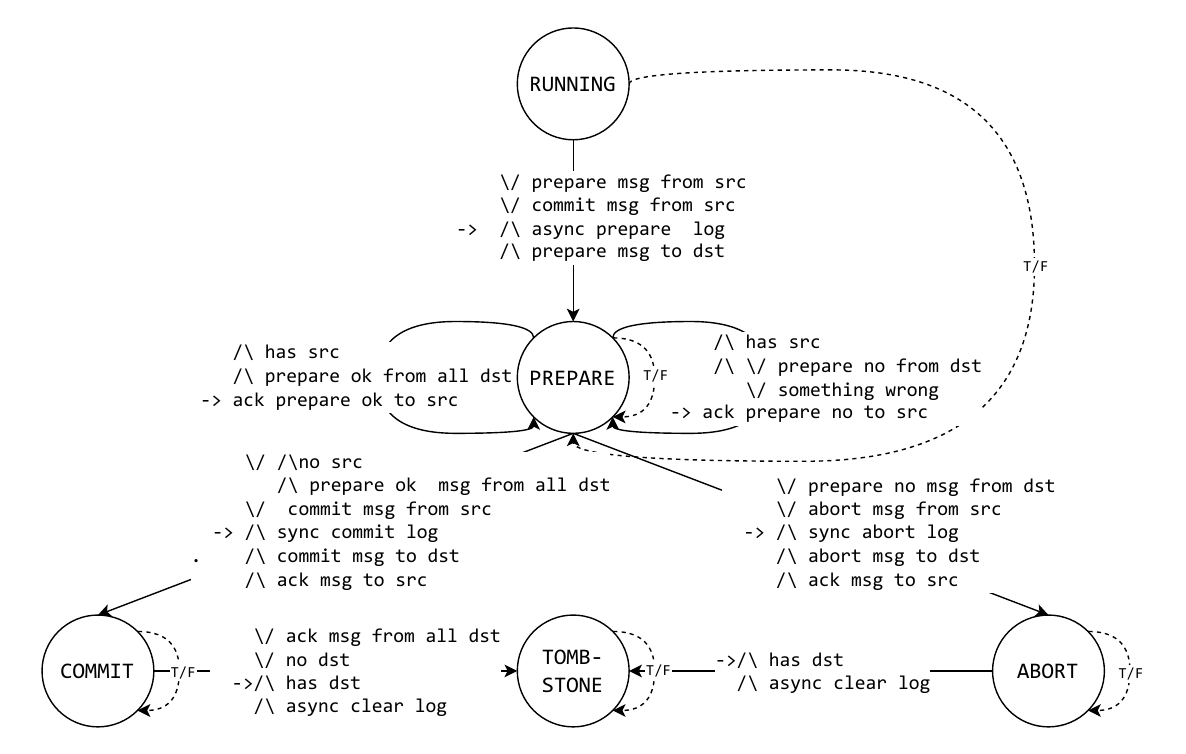}
  \caption{Tree-shaped two-phase commit state machine. Solid lines represent normal processes; dotted lines represent abnormal timeouts (T) and downtime (F). \texttt{/\textbackslash} represents AND, \texttt{\textbackslash/} represents OR. A $\rightarrow$ B means that in the current state, input A generates output B and transitions to the next state.}
  \label{fig:3}
\end{figure}

The state machine is illustrated in Figure~\ref{fig:3}. We describe the protocol in two phases:

\paragraph{Prepare Phase.} Initially, the node is in RUNNING state. Upon receiving a prepare message from its parent or a commit request from the scheduler, it asynchronously writes a prepare log, sends prepare messages to all children, and resets the \textit{collected participant list}. After the prepare log persists, the node enters PREPARE state. Upon receiving prepare OK/NO messages from children, it adds them to \textit{collected participant list} and sets \textit{status} to NO if any message is prepare NO.

\paragraph{Commit Phase.} When a node receives prepare OK from all children (i.e., \textit{collected participant list} equals the set of child nodes), it sends prepare OK to its parent (if it exists) or commit messages to all children (if it is the root). A node in PREPARE state receiving a commit message from its parent synchronously writes a commit log, sends commit messages to all children, resets \textit{collected participant list}, and sends ACK to its parent, entering COMMIT state. Finally, after receiving ACK from all children, it asynchronously writes a clear log and enters TOMBSTONE state.

The abort phase follows analogous logic. Detailed state machine specifications are provided in Appendix~\ref{sec:appendix:state_machine}.

\subsection{Summary}
\label{sec:tree_2pc:summary}

We summarize the performance characteristics of the tree-shaped 2PC algorithm under normal execution (assuming successful commit, $N$ is the average number of child participants per node, $H$ is the average tree height). The following metrics are for a given tree; scaling behavior (how latency and resource consumption evolve with cluster size or under partition transfer) is not analyzed in this section and is evaluated empirically in Section~\ref{sec:experimental_evaluation}.

\paragraph{Transaction Latency.} The time from when the user sends a COMMIT request until receiving a response. It requires $2H$ message roundtrips (PREPARE requests from coordinators/sub-coordinators and PREPARE responses from participants/sub-participants) and $1$ log synchronization (participants' PREPARE logs).

\paragraph{Row Lock Release Latency.} The time from when a lock is acquired until it is released, representing the conflict footprint. It requires $3H$ message roundtrips (PREPARE, COMMIT requests from coordinators/sub-coordinators, PREPARE responses from participants) and $2$ log synchronizations (participants' PREPARE and COMMIT logs).

\paragraph{Transaction Resource Consumption.} Total resources consumed by a transaction, in messages and logs: $5NH$ message transmissions (PREPARE, COMMIT, and RELEASE from coordinators; PREPARE and COMMIT responses from participants; see Section~\ref{sec:discussion:opt:release} for RELEASE) and $2NH$ synchronously replicated logs (participants' PREPARE and COMMIT logs).

\section{Transfer Process}
\label{sec:transfer}

\subsection{Transfer Principles}
\label{sec:transfer:principles}

In OceanBase 4.0, tablets~\cite{chang2008bigtable} and partitions serve as partitioning units: tablets are the smallest physical storage units, partitions are user-defined logical units. A partition may contain one or more tablets. Initially, partitions are distributed symmetrically across machines. As data changes, OceanBase uses partition transfer~\cite{atif2009analysis} to migrate partitions between servers for load balancing.

Section~\ref{sec:tree_2pc} presents the static tree-shaped 2PC protocol where the log stream tree remains unchanged. However, partition transfers can occur at any point during 2PC execution, dynamically altering the log stream tree (participants). To ensure ACID properties, we establish the following transfer principle. It is stated in terms of logs and messages; the two views are equivalent because 2PC messages are generated from and reflect the logged state (e.g., PREPARE/COMMIT messages correspond to the state persisted in PREPARE/COMMIT logs).

\begin{theorem}[Transfer principle] \label{thm:transfer_principle}
    \textbf{(Logs.)} Transaction logs written before the transfer-out log must be applied to the destination log stream through the transfer process. Transaction logs written after the transfer-out log must account for the transfer's impact (e.g., include the destination in the participant list).
    \textbf{(Messages.)} 2PC messages (PREPARE, COMMIT, ABORT) sent before transfer completion must be applied to the destination log stream through the transfer process. Messages sent after transfer completion must account for the transfer's impact. The log and message formulations are equivalent.
\end{theorem}
The proof is given in Appendix~\ref{sec:appendix:transfer_principle}.

\subsection{Dynamic Partitioning and Transfer}
\label{sec:transfer:con_transfer_transaction}

\begin{figure}
  \centering
  \includegraphics[width=\linewidth]{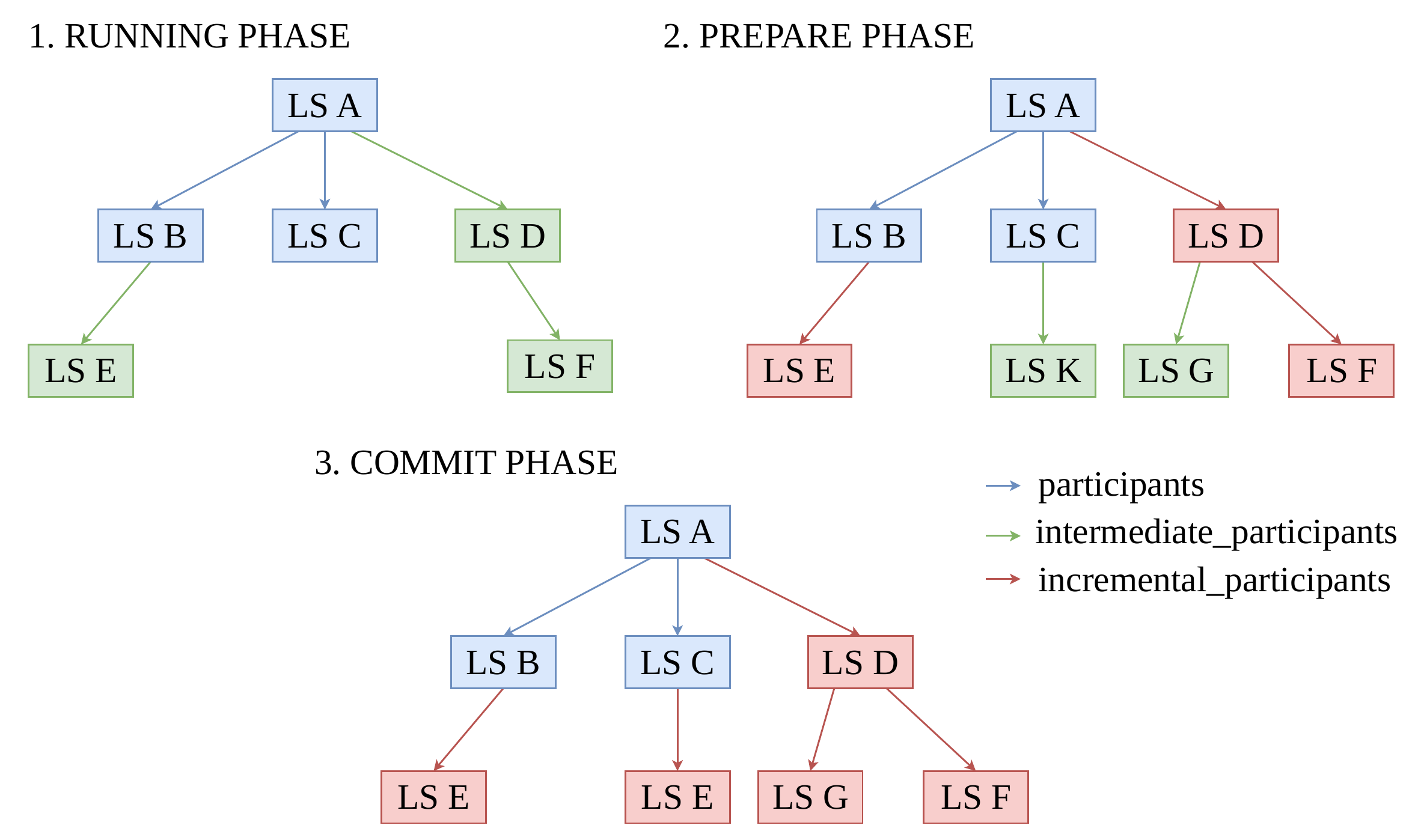}
  \caption{Transfer during tree-shaped 2PC.}
  \label{fig:transfer_2pc}
\end{figure}

To handle partition transfers during 2PC execution, we first consider an idealized model where partition transfer and 2PC log synchronization are atomic. To ensure complete participant lists at each stage per Theorem~\ref{thm:transfer_principle}, we maintain participants from previous stages for the current stage, and participants from the current stage for the next stage. The transaction context maintains three participant variables:

\begin{itemize}
    \item \textit{participants}: Participants from the SQL layer, persisted in both log stream and transaction context.
    \item \textit{incr\_parts}: Participants from partition transfers in the previous stage, included in the current stage's participant list, persisted in both log stream and transaction context.
    \item \textit{interm\_parts}: Participants from partition transfers in the current stage, included in the next stage's participant list, persisted only in transaction context.
\end{itemize}

Algorithm~\ref{alg:transfer_2pc} illustrates the commit process when transfers occur during RUNNING, PREPARE, and COMMIT stages (assuming successful commit). For each transfer $trans_i$, $trans_i.src$ and $trans_i.dst$ denote the source and destination log streams. The \textsc{migrateData} procedure ships transaction data and state from source to destination. The \textsc{commitLogs} procedure merges \textit{interm\_parts} into \textit{incr\_parts} and synchronizes 2PC logs containing \textit{participants} and \textit{incr\_parts} via Paxos. Example~\ref{example:transfer_2pc} demonstrates this process.

\begin{algorithm} 
\caption{Commit with Concurrent Partition Transfer}
\label{alg:transfer_2pc}
\begin{algorithmic}[1]
\Require
   Transaction $T$, scheduler $s$, log stream tree $tree$
\Statex \texttt{/* $T$ enters the RUNNING state */}
\State $s$ finds initial tablets and corresponding log streams $log_{initial}$. 
\State $tree \leftarrow log_{initial}$
\Statex \texttt{/* Transfer $trans_1$ occurs during RUNNING */}
\State $trans_1.src.interm\_parts.add(trans_1.dst)$
\State \textsc{migrateData($trans_1.src$, $trans_1.dst$)}
\Statex \texttt{/* $T$ enters PREPARE state */}
\State \textsc{commitLogs($tree$)}
\Statex \texttt{/* Transfer $trans_2$ occurs during PREPARE */}
\State $trans_2.src.interm\_parts.add(trans_2.dst)$
\State \textsc{migrateData($trans_2.src$, $trans_2.dst$)}
\Statex \texttt{/* $T$ enters COMMIT state */}
\State \textsc{commitLogs($tree$)}
\Statex \texttt{/* Transfer $trans_3$ occurs during COMMIT */}
\State \textsc{migrateData($trans_3.src$, $trans_3.dst$)}
\Statex 
\Procedure{\textsc{migrateData}}{$src$, $dst$}
    \State Migrate data and transaction state of $src$ context to $dst$.
\EndProcedure
\Statex
\Procedure{\textsc{commitLogs}}{$tree$}
    \For{all nodes $p$ \in $tree$}
        \State $p.incr\_parts.add(p.interm\_parts)$
        \State \textsc{Paxos($p.participants$, $p.incr\_parts$)} \Comment{Synchronize logs via Paxos}
    \EndFor
\EndProcedure
\end{algorithmic}
\end{algorithm}

\begin{example} \label{example:transfer_2pc}
    As shown in Figure~\ref{fig:transfer_2pc} (top left), initially the transaction is in RUNNING state with \textit{participants} A, B, and C (blue). During RUNNING, three partition transfers occur: A→D, B→E, and D→F. Consequently, D is in A's \textit{interm\_parts}, E is in B's \textit{interm\_parts}, and F is in D's \textit{interm\_parts} (green). When the transaction enters PREPARE (top right), \textit{interm\_parts} become \textit{incr\_parts} and join \textit{participants} in the current stage. Two transfers during PREPARE add K and G (green) as \textit{interm\_parts}. At COMMIT (bottom), all nodes in the log stream tree participate in the commit.
\end{example}

\subsection{Transfer and Log Commit Concurrency}
\label{sec:transfer_con_transfer_commit}

The atomic assumption in Section~\ref{sec:transfer:con_transfer_transaction} is idealized; in practice, partition transfer and 2PC log commits cannot be made atomic. For instance, when a transaction receives a PREPARE request and merges \textit{interm\_parts} into \textit{incr\_parts}, a concurrent partition transfer may add a destination log stream to \textit{interm\_parts} and commit the transfer-out log before the prepare log commits. This causes the prepare log to miss the transfer destination, violating Theorem~\ref{thm:transfer_principle}.

To ensure correctness, two requirements must be satisfied:
\begin{itemize}
    \item Partition transfer must migrate all transaction context written before the transfer-out log to the destination.
    \item Transactions must include all previous transfer destination log streams in \textit{incr\_parts}.
\end{itemize}

The first requirement is handled by the transfer mechanism. For the second, the following operations must be atomic (the shared transfer lock avoids races:
\begin{itemize}
    \item Partition transfer broadcasts the destination log stream to all transactions and commits the transfer-out log.
    \item Participants merge \textit{interm\_parts} into \textit{incr\_parts} and commit the 2PC log.
\end{itemize}

Algorithms~\ref{alg:partition_transfer} and~\ref{alg:log_commit} present the detailed partition transfer and log commit procedures.

\begin{algorithm} 
\caption{Partition Transfer Procedure}
\label{alg:partition_transfer}
\begin{algorithmic}[1]
\Require
   Transfer service $s$, transfer-out log $\mathit{log}$ (with tablet $\mathit{tab}$, destination log stream $\mathit{dst}$, timestamp $\mathit{ts}$, source log stream $\mathit{src}$), transfer lock $\mathit{lock}$, transaction contexts set $tc$
\State $s$ acquires $\mathit{lock}$
\For{transaction $t \in tc$}
    \If{$\mathit{tab}$ is involved in $t$}
        \State $t$ acquires $\mathit{lock}_t$
        \State Block $t$ from writing new logs
        \State Collect $t$'s transaction context into migration set (for destination log)
        \State $t$ releases $\mathit{lock}_t$
    \EndIf
\EndFor
\State $s$ commits log on source log stream: tablet $\mathit{tab}$, destination $\mathit{dst}$, timestamp $\mathit{ts}$
\State $s$ commits log on destination log stream: tablet $\mathit{tab}$, source $\mathit{src}$, transaction information set
\For{transaction $t \in tc$}
    \If{$\mathit{tab}$ is involved in $t$ \textbf{and} $\mathit{ts} > t.\mathit{ts}'$}
        \State $t$ acquires $\mathit{lock}_t$
        \State Add $\mathit{dst}$ into $t.\mathit{interm\_parts}$
        \State Unblock $t$ from writing logs
        \State $t$ releases $\mathit{lock}_t$
    \EndIf
\EndFor
\State $s$ releases $\mathit{lock}$
\State Clear the information on $s$
\end{algorithmic}
\end{algorithm}

\begin{algorithm} 
\caption{Log Commit Procedure}
\label{alg:log_commit}
\begin{algorithmic}[1]
\Require
   Transfer service $s$, transfer lock $\mathit{lock}$, transaction contexts set $tc$
\State $s$ acquires $\mathit{lock}$
\For{transaction $t \in tc$}
    \If{$t$ is not stopped from writing log}
        \State $t$ acquires $\mathit{lock}_t$
        \State Merge $t.\mathit{interm\_parts}$ into $t.\mathit{incr\_parts}$
        \State Commit logs containing $t.\mathit{participants}$ and $t.\mathit{incr\_parts}$ via Paxos
        \State $t$ releases $\mathit{lock}_t$
    \EndIf
\EndFor
\State $s$ releases $\mathit{lock}$
\end{algorithmic}
\end{algorithm}

\begin{theorem}\label{thm:transfer_correctness}
    Algorithms~\ref{alg:partition_transfer} and~\ref{alg:log_commit} ensure that the two requirements above are satisfied, and thus satisfy Theorem~\ref{thm:transfer_principle}.
\end{theorem}
The proof is given in Appendix~\ref{sec:appendix:transfer_correctness}.

\subsection{Exception Handling of Transfer}
\label{sec:transfer:exception}

Transfer must correctly handle participant lists in failure scenarios to ensure correctness. When migration succeeds but the source log stream fails, the correct participant list must be maintained. Transfer is itself implemented using transactions; therefore, in exception scenarios all transfer operations are rolled back, preserving consistency. Transaction state and partition state are updated atomically at commit time through an atomic replacement of the relevant state, so that either the full transfer outcome is visible or none of it is. Since transfer synchronizes executing transaction data to the destination log stream (Figure~\ref{fig:5}), we handle participant list maintenance through two mechanisms:

\textit{Migration-based transfer.} Transactions migrate via transaction logs or persistent transaction tables. During persistence, we update the transaction participant list, similar to master switching.

\textit{Double-write transfer.} We apply transfer logs to update participant lists. When a transaction creates a context on a log stream, it joins the participant list. By ensuring double-write completion (persisting all transaction logs in corresponding partitions before writing the transfer-out log), we can write transfer-out logs for transactions that survive in memory during failure recovery.

\begin{figure}
  \centering
  \includegraphics[width=\linewidth]{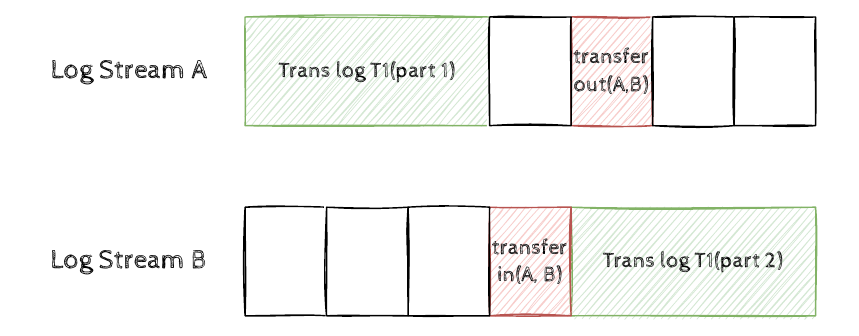}
  \caption{Log stream tree during execution.}
  \label{fig:5}
\end{figure}

\subsection{Transaction Dependence on Partitions}
\label{sec:transfer:dependence}

Tree-based 2PC relies on transfer to maintain correct participant lists (log stream trees containing the minimum set). The transfer process must update participant lists for transactions containing transferred partitions on the source side, requiring transactions to record included partitions. To minimize partition dependence, we set an upper limit on recorded partitions. When this limit is exceeded, transfers add destinations to the transaction's participant list. Since any participant list containing the minimum set is correct (Theorem~\ref{thm:log_stream_tree}), extra participants do not affect correctness.


\subsection{Circular Transactions}
\label{sec:transfer:circular}

Transactions may transfer across log streams and loop back to the original stream, requiring circular transaction handling \cite{alkhatib2002transaction, kolltveit2007circular}. When a participant receives multiple PREPARE messages from different senders, it does not immediately initiate handle\_2pc\_prepare\_request. Instead, it waits for its own log synchronization to complete (all nodes write logs per the state machine) but does not wait for downstream logs. It directly replies to its parent with PREPARE\_OK (downstream logs are managed by the same log stream in the parent to mitigate circular pipeline conflicts on linear commit progression). The modified state machine is:

\noindent$\bullet\ $ handle\_2pc\_prepare\_request($p$)
\begin{itemize}
    \item \textbf{Enable:} State is RUNNING and a 2PC\_PREPARE request is received from parent participant $p$, or a TXN\_COMMIT request is received from the scheduler.
    \item \textbf{Action:} If this is the first parent message, asynchronously write a prepare log, send 2PC\_PREPARE requests to all participants (including participant list and node status in logs and messages), and transition to PREPARE. Otherwise, wait for prepare log synchronization, send PREPARE response to parent $p$ (including participant list and node state), and transition to PREPARE.
\end{itemize}

\textbf{Liveness.} The modified rule preserves liveness and avoids deadlock despite circular topologies. Each node waits only for \emph{its own} prepare log to synchronize before replying PREPARE\_OK to its parent; it does not wait for any downstream or sibling node. Thus no node ever waits on another node that (directly or via the tree) waits on it, so there is no circular wait. Under the usual fairness assumptions (log synchronization completes in finite time and messages are eventually delivered), every participant that has received a PREPARE request will eventually complete its local prepare log sync and send PREPARE\_OK to its parent. Progress therefore flows from leaves toward the root, and the root can collect all votes and drive the transaction to COMMIT or ABORT in a finite number of steps. Hence the protocol eventually reaches a terminal state (all committed or all aborted), preserving the liveness guarantee of tree-shaped 2PC (cf.\ Section~\ref{sec:discussion:tla:liveness}).

\subsection{Achieving Design Goals}
\label{sec:transfer:goals}

This section demonstrates how we fulfill the two design goals from Section~\ref{sec:intro}: high performance for transactions without transfers, and simplicity for transactions with transfers.

\textit{High performance.} When a transaction commits with log streams that have no embedded participant lists, 2PC proceeds without partition awareness. With a log stream tree of height 1, latency and resource consumption match optimized 2PC at the same participant count, with improvement from using lightweight log streams as atomic units instead of partitions. With only one log stream and no participant list, a transaction simplifies to one-phase commit, achieving single-machine performance.

\textit{Simplicity.} Tree-based 2PC, as a classic protocol variant, requires minimal code changes by enabling intermediate nodes to act as both coordinators (downstream) and participants (upstream). In failure recovery, participant list immutability within a phase ensures straightforward transfer handling during commit, with critical logic centralized on managing transfer and transaction log ordering and processing, guaranteeing simplicity and robustness.

\section{``Unknown'' Transaction States}
\label{sec:unknown}

This section first describes the lying participants problem, then presents the unknown state mechanism to address it, and finally discusses how to eliminate the need for unknown states.

\subsection{The Lying Participants Problem}
\label{sec:unknown:lying}

In the tree-shaped 2PC protocol presented in Section~\ref{sec:tree_2pc}, participants (log streams) in the tree may ``lie'' to their parent coordinators when transaction contexts are lost. This problem can occur in any hierarchical 2PC structure, including our tree-shaped protocol. Example~\ref{example:lie} elaborates this scenario.

\begin{figure}
  \centering
  \includegraphics[width=\linewidth]{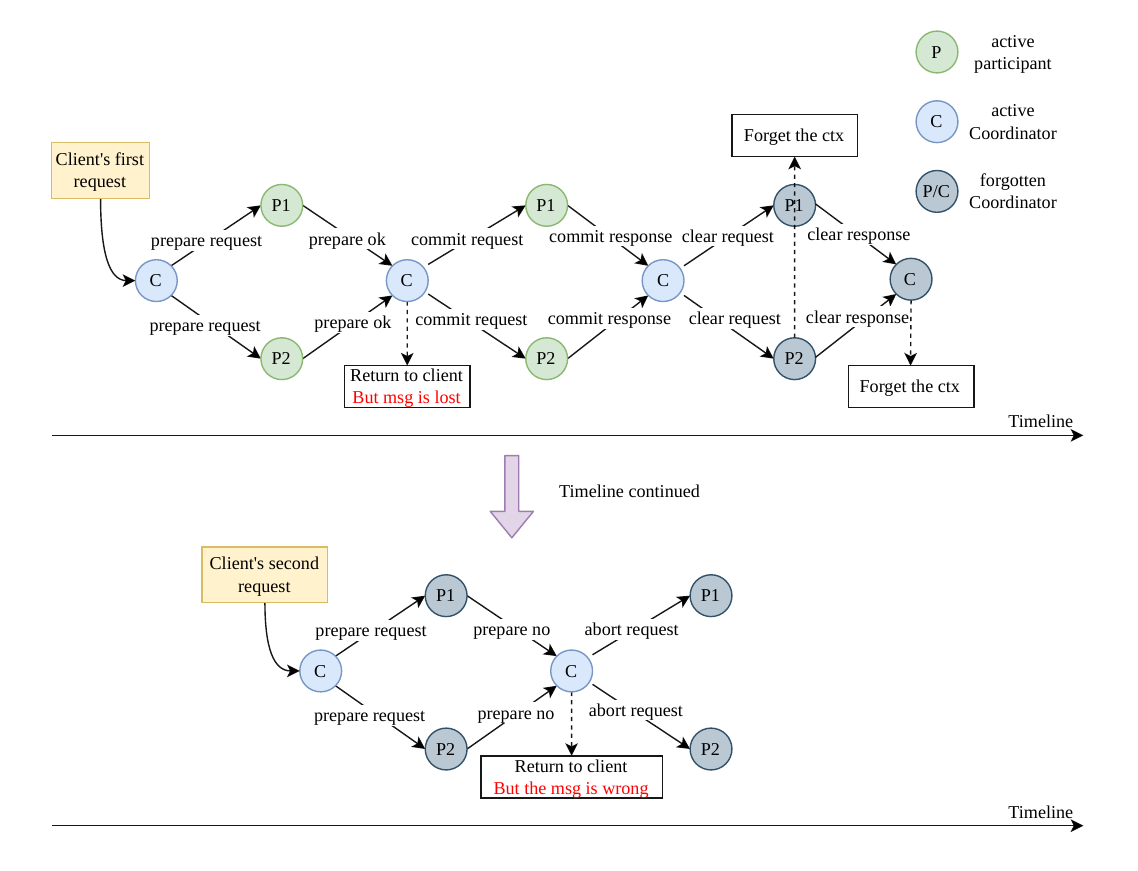}
  \caption{Issues caused by lying participants in tree-shaped 2PC.}
  \label{fig:3.3.3.1}
\end{figure}

\begin{example} \label{example:lie}
    As shown in Figure~\ref{fig:3.3.3.1} (top): Consider a tree-shaped 2PC structure with a root coordinator C and two participant log streams P1 and P2. After the coordinator C receives a COMMIT request and the tree-shaped 2PC protocol correctly transitions all participants (P1, P2) to the commit state, the transaction logs are reclaimed (note that the transaction context here refers only to the indices of the logs). However, the final COMMIT response to the user is lost. In this case, the user retries the transaction by re-sending the COMMIT request to re-establish the coordinator C's context. The tree-shaped 2PC process restarts (as shown in Figure~\ref{fig:3.3.3.1}, bottom). Since the participants' contexts have already been reclaimed, they can only respond with PREPARE\_NO. This forces the tree-shaped 2PC protocol to enter a rollback phase, ultimately returning an ABORT response to the user. Clearly, participants lied in their response to the second PREPARE request, leading the user to an incorrect state and violating Consistency (we treat the coordinator and user as special types of participants).
\end{example}

\subsection{prepare\_unknown and trans\_unknown}
\label{sec:unknown:states}


\begin{figure}
  \centering
  \includegraphics[width=\linewidth]{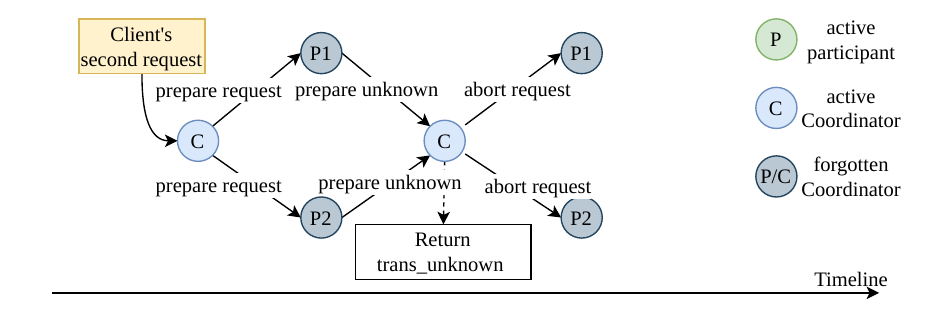}
  \caption{Resolving participant lying in tree-shaped 2PC with ``Unknown'' states.}
  \label{fig:3.3.3.2}
\end{figure}

To address the problem described above, prepare\_unknown is introduced in the tree-shaped 2PC state machine. When a participant (log stream) in the tree receives a PREPARE request from its parent coordinator but no transaction context exists (due to either normal commit cleanup or abnormal rollback), it responds with prepare\_unknown to its parent coordinator. This indicates that the participant no longer remembers its previous decision (e.g., whether it voted YES or NO in a prior phase).
Upon receiving prepare\_unknown, the coordinator proceeds by sending an ABORT request to all participants in its subtree, effectively terminating the transaction. 

For user responses, the logic becomes nuanced: if the root coordinator context is newly created and encounters prepare\_unknown, the system returns ABORT; if the root coordinator is recreated (e.g., after a lost response) and detects prepare\_unknown, it returns trans\_unknown. Revisiting the scenario in Figure~\ref{fig:3.3.3.1} (as shown in Figure~\ref{fig:3.3.3.2}), when the user retries the transaction, the participants in the log stream tree, having lost their context, respond with prepare\_unknown, prompting the root coordinator to return trans\_unknown. This avoids inconsistencies in the tree-shaped 2PC process and ensures consistency. Notably, the second coordinator is treated as distinct from the original—it never proceeds to commit, which does not violate consistency. However, the user must never observe inconsistency, so unknown states are used to handle such ambiguous scenarios.

\subsection{Strip ``Unknown'' States}
\label{sec:unknown:strip}

Although prepare\_unknown and trans\_unknown solve the problem of coordinators or participants lying to their callers, they significantly increase the complexity of tree-shaped 2PC. A log stream node may act as both a cascaded coordinator (for its children) and a participant (for its parent), so the unknown-state mechanism adds complexity to the state machine. We aim to reduce the scenarios where ``unknown'' states must be exposed to callers.

To avoid complex handling logic and to strip unknown states in common cases, we introduce a \emph{Transaction Data Table} (TDT) as in MaLT~\cite{malt_sigmod2025}. The TDT stores the decided state (committed or aborted) of every transaction that has reached a terminal state on this log stream. When a participant has lost its transaction context (e.g., after cleanup or recovery) and subsequently receives an inquiry message (e.g., a duplicate or retried PREPARE request), it first checks the TDT. If the transaction is found and already decided, the participant responds with the actual decision (e.g., PREPARE\_OK with the committed/aborted outcome) instead of replying prepare\_unknown. The coordinator can then proceed normally without entering the unknown-state path or returning trans\_unknown to the user. Transaction state in the TDT is retained for a configurable period (e.g., 30 minutes by default), which satisfies most duplicate or late messages (e.g., retries and delayed delivery). Thus, in the common case where the transaction has already committed or aborted and the inquirer is a duplicate or replayed message, the participant can resolve the inquiry from the TDT and avoid exposing ``unknown'' states, simplifying the state machine and reducing the need for user-facing trans\_unknown handling.

\section{Discussion}
\label{sec:discussion}

This section discusses several important aspects of our tree-shaped 2PC framework. We first present optimizations to improve transaction performance, including row lock release latency reduction and CLEAR log minimization. We then provide formal verification of the protocol's safety and liveness properties using TLA+. Finally, we discuss integration considerations with existing OceanBase modules, including liboblog and backup recovery mechanisms.

\subsection{Optimization}
\label{sec:discussion:opt}

\subsubsection{Release Messages}
\label{sec:discussion:opt:release}

In our tree-shaped 2PC implementation, we observed that while transaction latency remains unchanged and resource consumption is reduced, the latency for resource release has increased by the duration of one log operation. Among these resources, row locks and logs are two critical components~\cite{ryu1990analysis, gray2005notes, barthels2019strong}: row locks represent a transaction's contention footprint, while logs govern log reclamation. While log reclamation latency is inherently tolerable and can be mitigated through delayed batch processing, row lock contention is intolerable: higher contention footprints directly degrade transaction throughput.

To address this issue, we propose an optimization to reduce row lock release latency. The coordinator's synchronous write of the COMMIT log ensures participants can safely release their log resources by querying the coordinator's transaction status if other participants fail, guaranteeing no residual dependencies between participants. However, row lock release does not require this synchronous log operation. Once any participant confirms the transaction's successful commit, row locks can be safely released.

We decouple row lock release and log reclamation using two distinct messages. After the coordinator enters the commit point (i.e., receiving all PREPARE acknowledgments), it concurrently initiates RELEASE messages (for row locks) and synchronizes the COMMIT log. Only after COMMIT log synchronization succeeds does it send COMMIT messages (triggering log resource cleanup). This maintains log reclamation latency at two log synchronizations while reducing row lock release latency by replacing one log synchronization with an additional message roundtrip.

The following summarizes the optimized tree-shaped 2PC algorithm (assuming successful commit with N participants):
\begin{itemize}
    \item \textbf{Transaction latency}: 2 message roundtrips (coordinator's PREPARE request, participants' PREPARE responses) and 1 log synchronization (participants' PREPARE logs). 
    \item \textbf{Row lock release latency}: 3 message roundtrips (coordinator's PREPARE, RELEASE requests; participants' PREPARE responses) and 1 log synchronization (participants' PREPARE logs).
    \item \textbf{Log reclamation latency}: 3 message roundtrips (coordinator's PREPARE, COMMIT requests; participants' PREPARE responses) and 2 log synchronizations (coordinator's COMMIT log, participants' PREPARE and COMMIT logs).
    \item \textbf{Transaction resource consumption}: 5N message transmissions (comprising the coordinator's PREPARE, RELEASE, and COMMIT requests, along with participants' PREPARE and COMMIT responses), 2N+1 synchronously replicated majority-acknowledged logs (the coordinator's COMMIT log and participants' PREPARE and COMMIT logs), and 1 asynchronously replicated majority-acknowledged log (the coordinator's CLEAR log).
\end{itemize}

\subsubsection{Optimizing CLEAR Logs via D2PC}
\label{sec:discussion:optimization:D2PC}

\begin{figure}
  \centering
  \includegraphics[width=\linewidth]{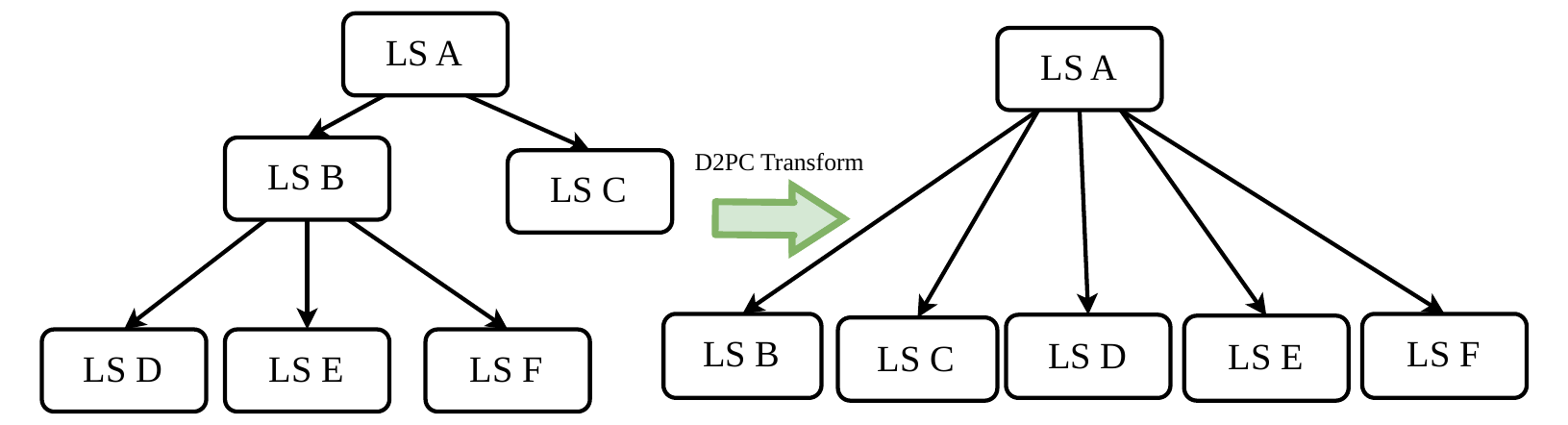}
  \caption{D2PC optimization.}
  \label{fig:7}
\end{figure}

We observe that only leaf nodes do not require writing CLEAR logs, which may introduce inefficiencies under our optimization framework. Notably, the participant lists in the PREPARE and COMMIT phases can differ due to partition transfers. To address the excessive CLEAR logs caused by frequent partition transfers during the two-phase commit process, we leverage the D2PC (Dynamic Decentralized Two-Phase Commit) method~\cite{ezhilchelvan2018non}. As illustrated in Figure~\ref{fig:7}, this approach ensures that only the root log stream (Log Stream A, i.e., LS A) needs to write CLEAR logs. By dynamically adjusting the participant list between phases, redundant logging is minimized while preserving atomicity and consistency guarantees.

\subsection{TLA+ Proof}
\label{sec:discussion:tla+}

In this section, we prove the safety and liveness properties of the proposed tree-shaped 2PC protocol. 
The protocol is specified in TLA+ (Appendix~\ref{sec:appendix:tla_spec}, module \texttt{2pc\_tla}). 
Unlike the previous flat TM/RM abstraction, the specification models each node uniformly: 
a node may act as a coordinator toward its children and simultaneously as a participant toward its parent. 
The local state of a node is recorded in \textit{rmState}, 
and the parent–child topology is captured by the variables \textit{parent} and \textit{children}. 
Dynamic membership introduced by transfer is represented by an intermediate child set, 
which is merged into the active child set only at phase boundaries. 
All protocol actions (prepare, decision, propagation, acknowledgement, internal abort, and transfer) 
are defined in the next-state relation of the specification.

\subsubsection{Proof of Safety}
\label{sec:discussion:tla:safety}

For safety, we must ensure that all nodes always agree on the same outcome and that once a node reaches commit or abort, it never changes that decision. 

\begin{theorem}\label{thm:safety_correctness}
From the TLA+ specification, the safety property of the proposed tree-shaped 2PC protocol follows directly from the safety property of the traditional 2PC protocol.
\end{theorem}
The proof is given in Appendix~\ref{sec:appendix:safety}.

\subsubsection{Proof of Liveness}
\label{sec:discussion:tla:liveness}

For liveness, we require that after the root initiates the protocol, all nodes eventually reach a terminal state (commit or abort) within a finite number of transitions.

\begin{theorem}\label{thm:liveness_correctness}
From the TLA+ specification, the liveness of the proposed tree-shaped 2PC protocol follows from that of the traditional 2PC protocol under standard fairness assumptions.
\end{theorem}
The proof is given in Appendix~\ref{sec:appendix:liveness}.

\section{Experimental Evaluation}
\label{sec:experimental_evaluation}

This section presents an experimental evaluation of our tree-structured 2PC framework in OceanBase, with a particular focus on its behavior under partition transfer.

\subsection{Experimental Setup}
\label{sec:experimental_setup}

This subsection details the experimental environment and methodology used for our evaluation. All experiments are conducted on OceanBase version 4.4.1 (community edition) in shared-nothing (SN) deployment mode. To comprehensively assess the proposed design, we employ standard OLTP workloads including TPC-C and sysbench. Unless otherwise specified, all experiments follow the uniform deployment configurations, workload settings, and evaluation metrics described below.

\subsubsection{System Configuration}
\label{sec:system_configuration}

We deploy a cluster with up to eight OBServer nodes and a single proxy node, where all OBServer nodes have identical hardware configurations. Each OBServer node is equipped with an Intel Xeon CPU with 32 cores, 128\,GB of memory, and a 2\,TB SSD. The cluster runs on Ubuntu 24.04 and is deployed in a cloud environment on Alibaba Cloud.
Each log stream is replicated using a two-replica Paxos configuration, with replicas placed on different OBServer nodes. We evaluate both transfer-enabled and transfer-disabled settings, and apply customized transaction-related configurations where necessary. Unless otherwise stated, the same system configuration is used across all experiments.

\subsubsection{Workloads}
\label{sec:workloads}

We use the TPC-C benchmark to evaluate the impact of partition transfer on online transaction processing under high concurrency. TPC-C models a realistic multi-table OLTP workload with a high proportion of distributed transactions, making it well suited for examining how partition transfer affects transaction latency, throughput, and performance stability. In our experiments, partition transfer is periodically injected during the steady-state execution phase to emulate continuous load balancing in elastic deployment environments.
Unless otherwise specified, TPC-C is configured with 3,000 warehouses and 500 concurrent client connections. Each experiment runs for one hour, consisting of an initial warm-up phase followed by steady-state measurement. Partition transfer is triggered at a fixed interval of three minutes during the steady-state phase.

\subsubsection{Metrics}
\label{sec:metrics}

To evaluate the performance and robustness of our framework, we employ two primary categories of metrics: throughput and latency. For the TPC-C benchmark, throughput is measured in transactions per minute (\textit{tpmC}). Transaction latency is captured from the client side, where we report both the average latency and the 99th percentile (P99) latency. We specifically prioritize P99 as the key indicator of tail latency, as it effectively reflects the impact of distributed coordination overhead and potential stalls during partition transfers.
Furthermore, we assess system stability by monitoring latency fluctuations throughout the execution. Rather than merely reporting aggregate statistics, we analyze the distribution of P99 spikes during active partition transfer events. This allows us to quantify the framework's resilience and its ability to maintain predictable performance despite dynamic changes in the underlying log stream topology.

\subsection{Impact of Transfer on Online Transactions}
\label{sec:impact}
In this subsection, we evaluate the impact of partition transfer on online transaction processing performance. We focus on how transfer affects transaction latency, throughput, and performance stability under realistic high-concurrency workloads.

\subsubsection{Experimental Setup and Baselines}
\label{sec:baselines}
We use the TPC-C benchmark to evaluate the impact of partition transfer on online transaction performance. To isolate the specific effects of dynamic topology changes, we focus on two primary configurations of our framework:
\begin{itemize}
    \item \textbf{OB (Transfer Disabled):} Partition transfer is deactivated throughout the experiment to establish a baseline for peak performance.
    \item \textbf{OB-Transfer (Transfer Enabled):} Partition transfer is periodically triggered at fixed three-minute intervals during the steady-state phase, simulating load balancing.
\end{itemize}

To provide broader performance context, we compare these configurations against two external baselines: MySQL (v8.0), a standalone database, and TiDB~\cite{huang2020tidb} (v8.1), a distributed database using traditional 2PC coordination. Unless otherwise stated, workload parameters (e.g., 3,000 warehouses) remain identical across settings.

\subsubsection{Overall Performance under Transfer}
\label{sec:overall-performance-transfer}

\begin{figure}[t]
  \centering
   \includegraphics[width=\linewidth]{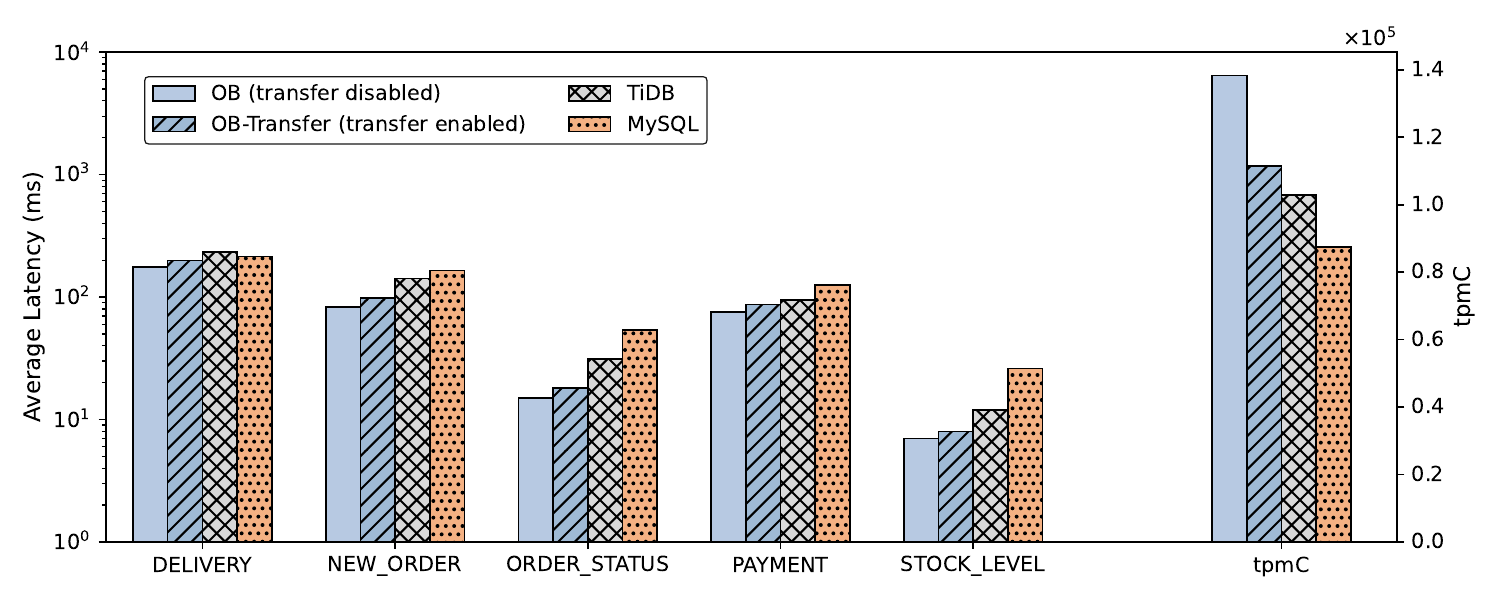}
  \caption{TPC-C latency--throughput comparison under partition transfer.}
  \label{fig:tpcc-transfer-performance}
\end{figure}

Figure~\ref{fig:tpcc-transfer-performance} presents the latency--throughput characteristics of TPC-C under partition transfer. The results reveal several key insights into the efficiency of our framework under both static and dynamic conditions.

\paragraph{Throughput and System Resilience}
As illustrated in the right panel of Figure~\ref{fig:tpcc-transfer-performance}, the baseline OceanBase (OB) achieves the highest throughput among the evaluated systems. When partition transfer is enabled (OB-Transfer), throughput decreases moderately due to the additional coordination required during log stream migration. However, the system sustains stable throughput without abrupt degradation, indicating that the transfer process is well integrated with transaction execution.

This behavior can be attributed to the tree-structured 2PC protocol, which allows the participant set of a transaction to evolve during execution without forcing transaction aborts or global synchronization barriers. By organizing transaction participants in a hierarchical structure, the protocol limits the coordination scope introduced by transfer, ensuring that the additional overhead grows in a controlled manner rather than propagating across all participants.

In contrast, TiDB exhibits lower throughput under the same workload configuration, reflecting the higher coordination overhead incurred by distributed transaction execution across multiple nodes. MySQL achieves substantially lower throughput in this high-concurrency setting. This outcome is primarily due to the saturation of single-node resources under 500 concurrent client connections, where CPU scheduling, lock management, and log serialization contend for limited shared resources. As concurrency increases, these contention points become dominant, leading to rapid throughput saturation rather than gradual scaling.

\paragraph{Transaction Latency Characteristics and Root Causes}
The left panel of Figure~\ref{fig:tpcc-transfer-performance} reports client-side latency for individual TPC-C transaction types on a logarithmic scale. OceanBase consistently achieves lower latency than TiDB across all transaction categories, indicating the effectiveness of consolidating multiple partitions within a single log stream. This Single Log Stream (SLS) optimization reduces the number of transaction participants involved in the commit process, thereby minimizing RPC invocations and cross-node coordination overhead.

A key observation is the close proximity between the OB and OB-Transfer latency bars. Even for coordination-intensive transactions such as \textit{New-Order} and \textit{Payment}, enabling transfer introduces only a limited latency increase. This is because the tree-structured commit protocol confines the impact of transfer-induced topology changes to a small portion of the commit path, preventing widespread coordination delays and preserving low end-to-end latency.

Interestingly, under the evaluated high-concurrency workload, the distributed deployment of OceanBase also provides more favorable latency profiles than the standalone MySQL baseline. By distributing client connections and transaction execution across multiple OBServer nodes, OceanBase alleviates contention on centralized resources. In contrast, MySQL processes all transactions within a single execution context, where increased concurrency amplifies lock contention and logging bottlenecks, resulting in higher latency and reduced throughput.

Overall, these results demonstrate that the proposed framework not only optimizes distributed transaction coordination in stable execution states, but also effectively absorbs the additional complexity introduced by partition transfer, enabling predictable and efficient online transaction processing in elastic environments.

\begin{figure}[t]
  \centering
   \includegraphics[width=\linewidth]{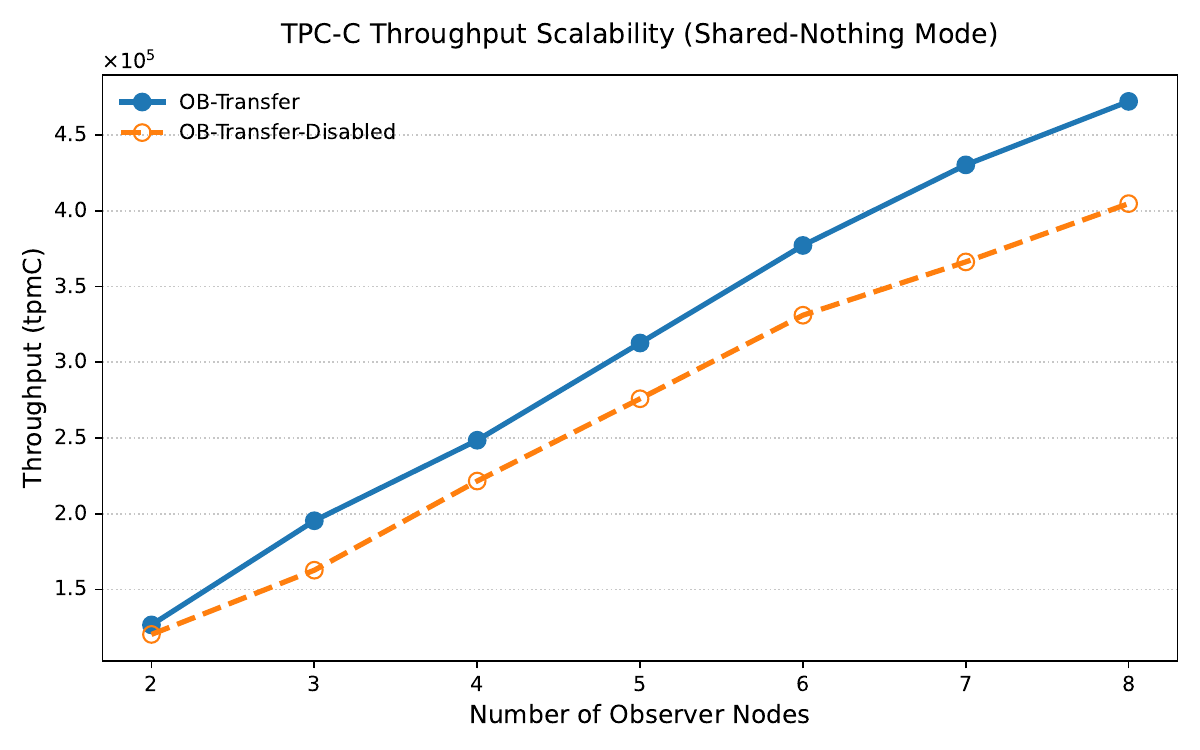}
  \caption{TPC-C throughput scalability as a function of cluster size.}
  \label{fig:tpmc-scalability}
\end{figure}

\begin{figure}[t]
  \centering
   \includegraphics[width=\linewidth]{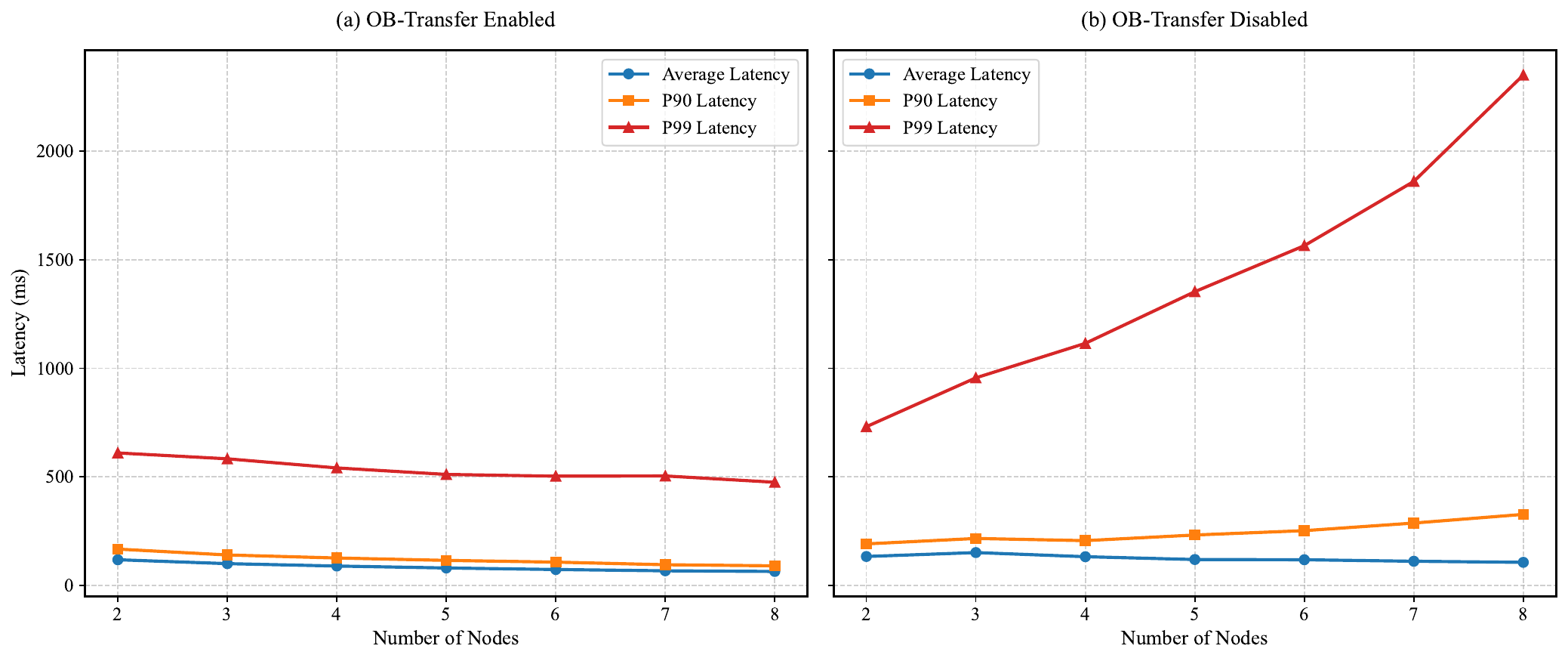}
  \caption{Transaction latency distribution across different cluster sizes.}
  \label{fig:tpmc-scalability-latency}
\end{figure}

\subsection{Scalability under Partition Transfer}
\label{sec:scalability-transfer}

This subsection evaluates the scalability of partition transfer under cluster expansion in shared-nothing (SN) mode. The primary objective is to verify whether the proposed transfer-aware transaction framework preserves high throughput and stable latency as the number of OBServer nodes increases.

\subsubsection{Experimental Methodology}
\label{sec:experimental_methodology}

To evaluate scalability, we varied the number of OBServer nodes in the OceanBase cluster while keeping the workload and system configuration unchanged from Section~\ref{sec:experimental_setup}. All experiments are conducted in SN mode with partition transfer enabled. We use the TPC-C benchmark configured with 3,000 warehouses and 500 concurrent client connections.

We scale the cluster from 2 to 8 OBServer nodes and evaluate system performance at each intermediate cluster size. For each configuration, we measure the achieved throughput in transactions per minute (tpmC) as well as client-side transaction latency. This experimental design isolates the impact of cluster expansion on transfer behavior and transaction processing performance.

\subsubsection{Throughput Scalability under Cluster Expansion}
\label{sec:throughput_scalability}

Figure~\ref{fig:tpmc-scalability} illustrates TPC-C throughput as a function of the number of OBServer nodes. The results demonstrate that throughput exhibits \textit{near-linear scalability} in both transfer-disabled and transfer-enabled configurations, indicating effective resource utilization.

Notably, while the transfer-enabled configuration shows slightly lower throughput due to the inherent coordination and logging overhead of partition migration, the performance delta between the two configurations remains constant as the cluster scales. The throughput curves exhibit parallel growth trajectories, suggesting that the overhead introduced by partition transfer is bounded and does not compound with cluster size.

This scalability is a direct result of our log-stream-based architecture and the tree-structured 2PC protocol. By aggregating co-located partitions into a single log stream participant, the protocol bounds the number of participants involved in a transaction. Furthermore, partition transfers only introduce additional participants along localized branches of the commit tree rather than expanding the global coordination scope. Consequently, scaling out the cluster increases aggregate processing capacity without proportionally increasing the per-transaction coordination cost.

\subsubsection{Latency Stability across Different Cluster Sizes}
\label{sec:latency_stability}

Figure~\ref{fig:tpmc-scalability-latency} reports the distribution of transaction latency (Average, P90, and P99) across different cluster sizes. The results indicate that transaction latency remains stable as the number of OBServer nodes increases.

Crucially, the tail latency (P99) exhibits no significant degradation during scale-out, confirming that the coordination overhead related to transfers remains bounded. Although partition transfer necessitates additional logging and message exchanges, these costs are confined to the specific log streams involved in the transaction's commit path. They do not accumulate with the cluster size. These findings confirm that our framework effectively decouples scalability from dynamic partition management, ensuring predictable performance in elastic, shared-nothing environments.

\subsection{Transfer Efficiency and Overhead}
\label{sec:transfer-efficiency}

In this subsection, we evaluate the execution efficiency and resource overhead of partition transfer under high-concurrency transactional workloads. The goal of this experiment is to examine whether partition transfer remains lightweight and predictable when executed concurrently with distributed transactions, rather than becoming a performance or resource bottleneck.

\subsubsection{Transfer Efficiency Experimental Setup}
\label{sec:transfer_efficiency_setup}

We conduct the experiment on an OceanBase cluster with partition transfer enabled. The workload is generated using sysbench with distributed read-write transactions. The number of concurrent client threads is varied from 64 to 1024 to emulate increasing transactional pressure. Each sysbench transaction includes point queries and delete--insert operations and is executed within a transactional context.

Partition transfer is performed at table granularity. The cluster contains 30 tables; each transfer randomly selects one table and migrates it from its current log stream to another randomly chosen log stream. This strategy avoids bias toward specific tables or access patterns and reflects realistic transfer scheduling. We evaluate two scenarios: \emph{intra-zone transfer}, where source and destination log streams reside in the same zone, and \emph{cross-zone transfer}, where the transfer spans different zones.

\subsubsection{Transfer Execution Time}
\label{sec:transfer_execution_time}

\begin{figure}[t]
  \centering
   \includegraphics[width=\linewidth]{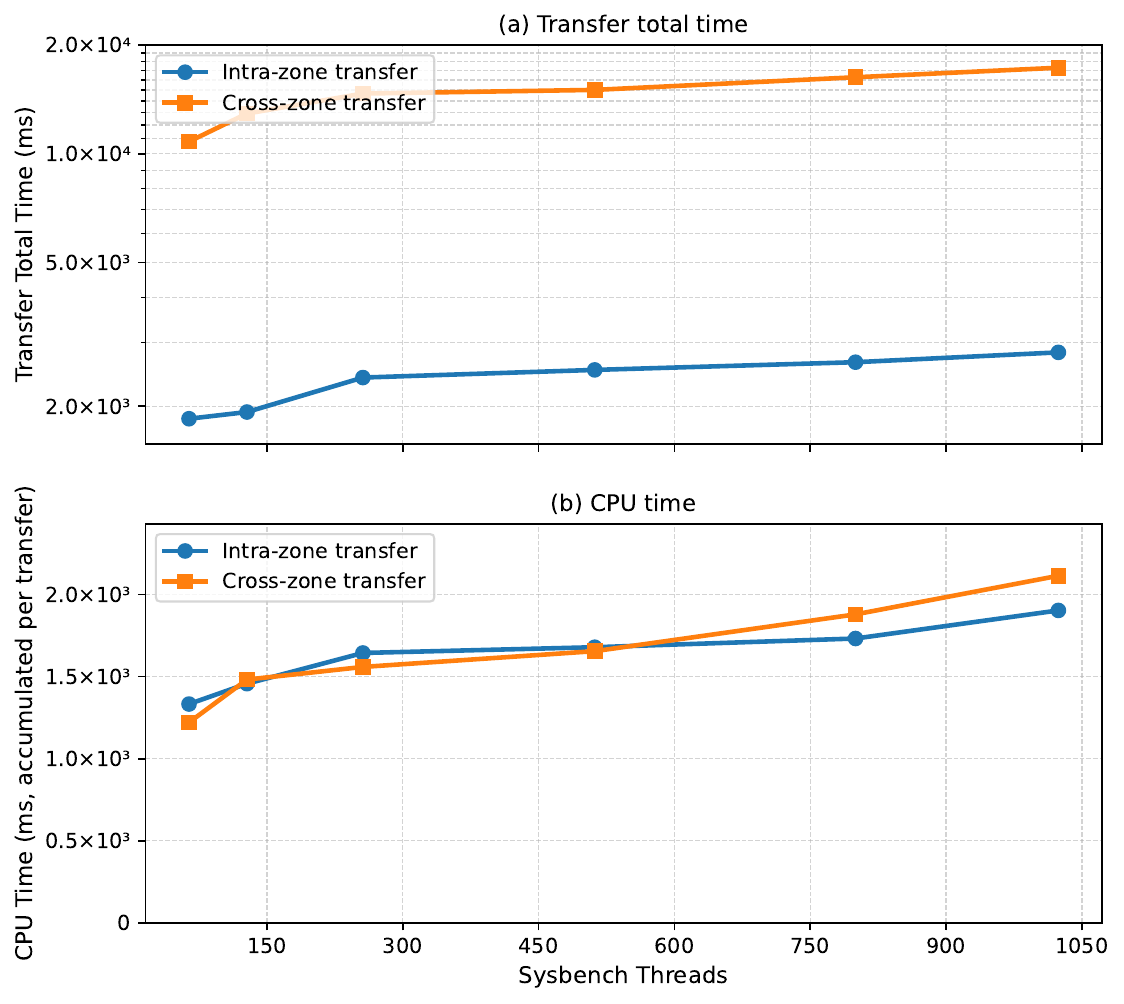}
  \caption{Partition transfer execution time and CPU overhead under varying sysbench thread counts.}
  \label{fig:transfer-time-cpu}
\end{figure}

Figure~\ref{fig:transfer-time-cpu}~(a) reports the total execution time of partition transfer under different sysbench thread counts. As transactional pressure increases, the transfer execution time grows slowly and remains stable across all tested configurations. Even under high concurrency, transfer latency does not exhibit abrupt increases or instability.

Cross-zone transfers incur higher latency than intra-zone transfers due to the additional communication and coordination required across zones. However, the performance gap between the two scenarios remains consistent as transactional load increases. The nearly parallel trends indicate that the overhead introduced by cross-zone communication is bounded and does not scale with the intensity of concurrent transactions. Taken together, these results indicate that partition transfer execution is largely decoupled from foreground transaction pressure and does not interfere with transaction processing on the critical path.

\subsubsection{CPU Overhead of Partition Transfer}
\label{sec:cpu_overhead}

Figure~\ref{fig:transfer-time-cpu}~(b) shows the accumulated CPU time consumed per transfer under different sysbench thread counts. CPU time is measured by aggregating the execution costs of transfer-related operations, including transaction filtering, transfer log generation, metadata reconstruction, transaction migration, and the physical transfer phase.

The results show that CPU consumption per transfer increases moderately with higher transactional pressure but remains within a narrow, predictable range. As with transfer latency, cross-zone transfers consume more CPU resources than intra-zone transfers, yet the difference remains stable across load levels. The bounded CPU overhead shows that partition transfer does not become a CPU hotspot even under high-concurrency workloads, confirming that the transfer mechanism is lightweight and can run alongside intensive distributed transaction processing.

\subsection{Tree-shaped Commit Performance}
\label{sec:exp:tree_commit}

This subsection evaluates how the geometry of the commit tree affects transaction performance under partition transfer.
Our tree-structured 2PC protocol organizes log streams as hierarchical participants, and transfer may embed additional log streams along localized branches of the commit tree.
As analyzed in Section~\ref{sec:tree_2pc:summary}, the critical-path cost of tree-shaped 2PC is primarily determined by the tree height $H$: transaction response latency requires $2H$ message
round-trips and one log synchronization, while lock release latency requires $3H$
round-trips and two log synchronizations.
In contrast, increasing fan-out mainly adds parallel branches without extending the
number of coordination rounds.

Following the evaluation methodology of Section~\ref{sec:experimental_setup}, we report both throughput and client-side latency metrics.
We emphasize tail latency as a key indicator of coordination stalls and transfer-induced instability, consistent with our prior evaluations that prioritize P99 behavior under dynamic transfer events.

\subsubsection{Methodology: Constructing Width/Depth-controlled Commit Trees}
\label{sec:exp:tree_commit:method}

We construct commit trees with controlled \emph{width} and \emph{depth} by manipulating transfer patterns while keeping workload intensity and system resources unchanged.
To increase \textbf{tree depth}, we repeatedly transfer the same partition across different log stream groups, forming a transfer chain that extends the commit path (i.e., increases $H$).
To increase \textbf{tree width}, we transfer multiple partitions to different log stream groups at the same hierarchy level, enlarging the fan-out while keeping the depth shallow.
We then measure the achieved throughput (TPS) and transaction latency (mean, P90, P99) under \texttt{transfer-on} and the \texttt{transfer-off} baseline.

\subsubsection{Impact of Tree Width}
\label{sec:exp:tree_commit:width}

\begin{figure*}[t]
  \centering
  \begin{subfigure}[t]{0.48\textwidth}
    \centering
    \includegraphics[width=\linewidth]{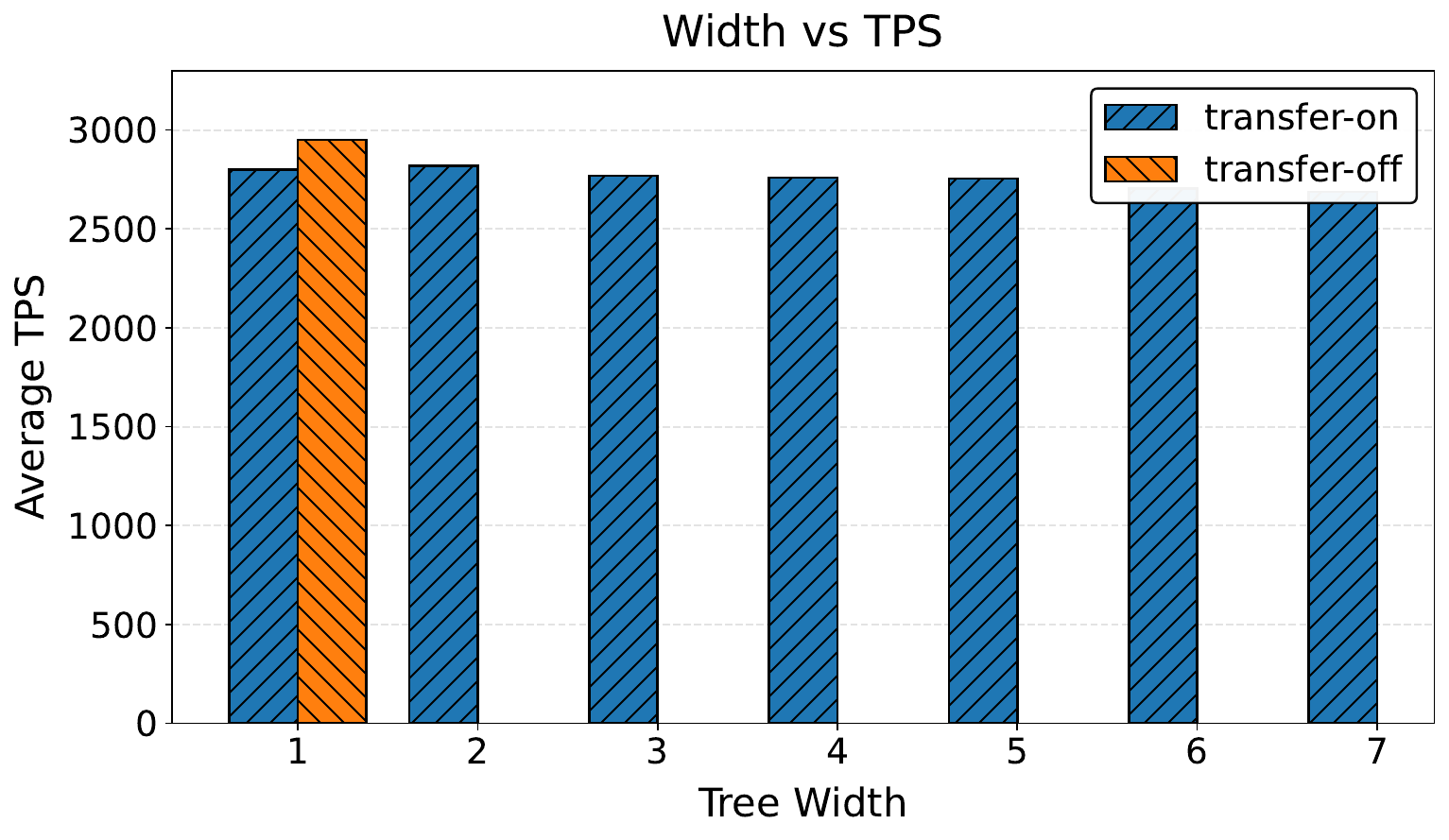}
    \caption{Throughput vs.\ tree width.}
    \label{fig:exp_width_vs_tps}
  \end{subfigure}
  \hfill
  \begin{subfigure}[t]{0.48\textwidth}
    \centering
    \includegraphics[width=\linewidth]{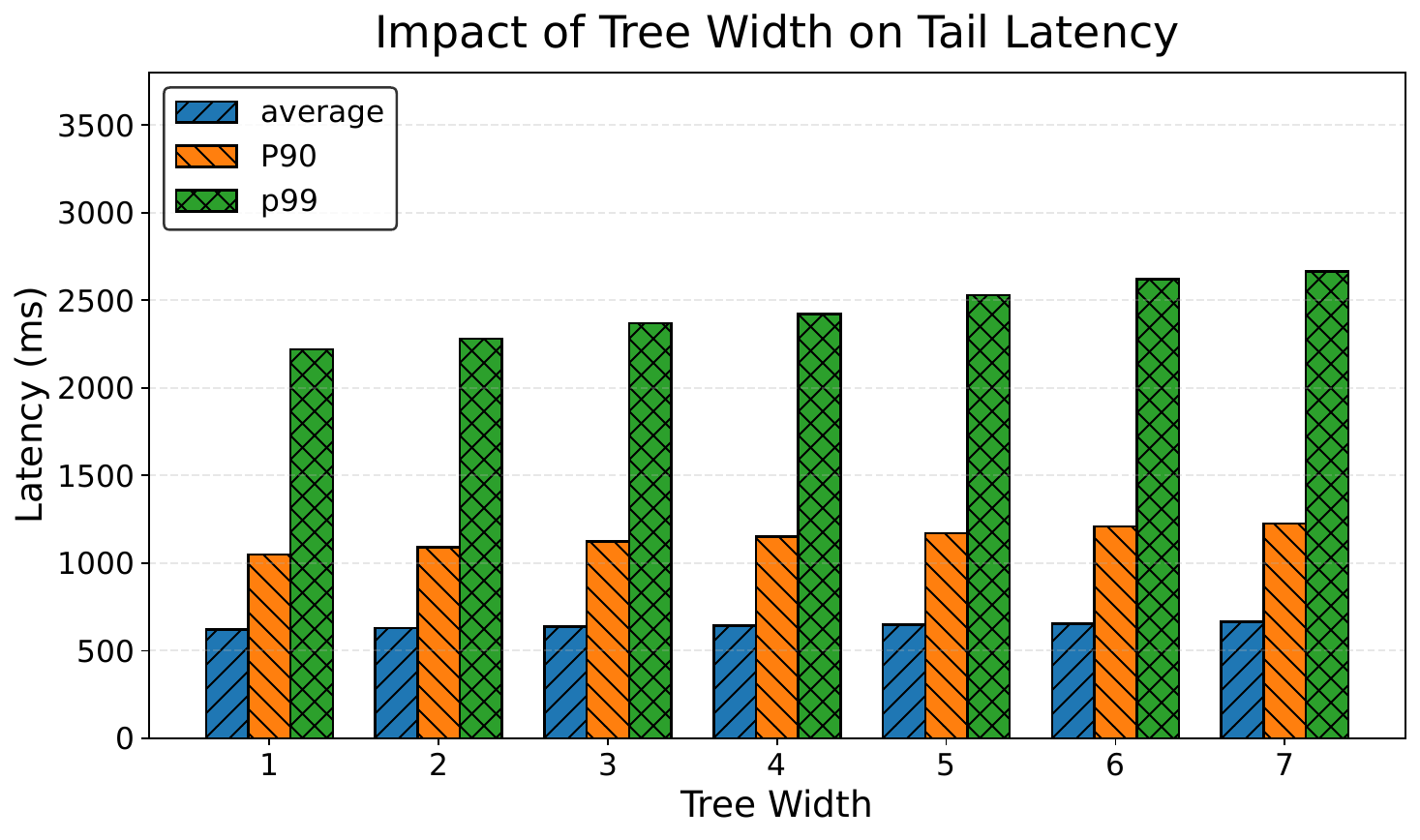}
    \caption{Latency vs.\ tree width.}
    \label{fig:exp_width_vs_latency}
  \end{subfigure}
  \caption{Tree width sensitivity under transfer.}
  \label{fig:exp_width_results}
\end{figure*}

Figure~\ref{fig:exp_width_vs_tps} reports throughput as the tree width increases.
We observe that TPS remains stable across a wide range of widths under \texttt{transfer-on}, with only minor fluctuations.
This indicates that increasing fan-out does not significantly affect the critical commit path.
The result aligns with the analysis in Section~\ref{sec:tree_2pc:summary}: width primarily increases the number of parallel branches, while the number of coordination rounds on the critical path remains unchanged.

Figure~\ref{fig:exp_width_vs_latency} shows the corresponding latency distributions.
Both mean and P90 latency increase only slightly as width grows, suggesting that the average coordination overhead remains bounded.
In contrast, P99 exhibits a moderate upward trend with larger widths.
This behavior is expected in hierarchical coordination: a larger fan-out increases the probability of encountering a straggler branch (e.g., transient queuing, log flush delay, or scheduling delay) among parallel participants, which is reflected primarily in tail latency.
Importantly, the P99 increase is gradual rather than abrupt, indicating that width-induced overhead does not amplify into system-wide instability; instead, it remains localized to the affected commit branches.

\subsubsection{Impact of Tree Depth}
\label{sec:exp:tree_commit:depth}

\begin{figure*}[t]
  \centering
  \begin{subfigure}[t]{0.48\textwidth}
    \centering
    \includegraphics[width=\linewidth]{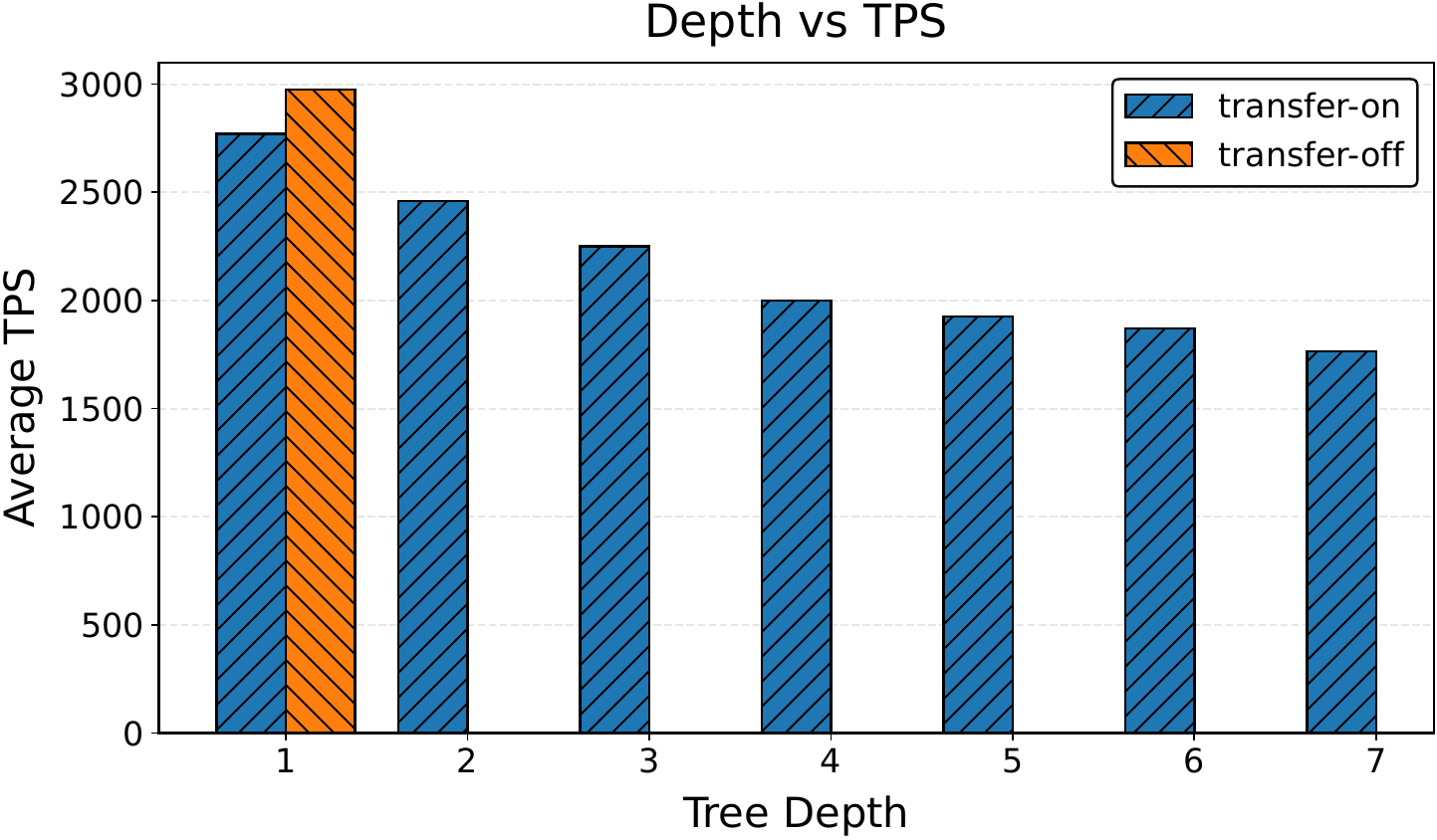}
    \caption{Throughput vs.\ tree depth.}
    \label{fig:exp_depth_vs_tps}
  \end{subfigure}
  \hfill
  \begin{subfigure}[t]{0.48\textwidth}
    \centering
    \includegraphics[width=\linewidth]{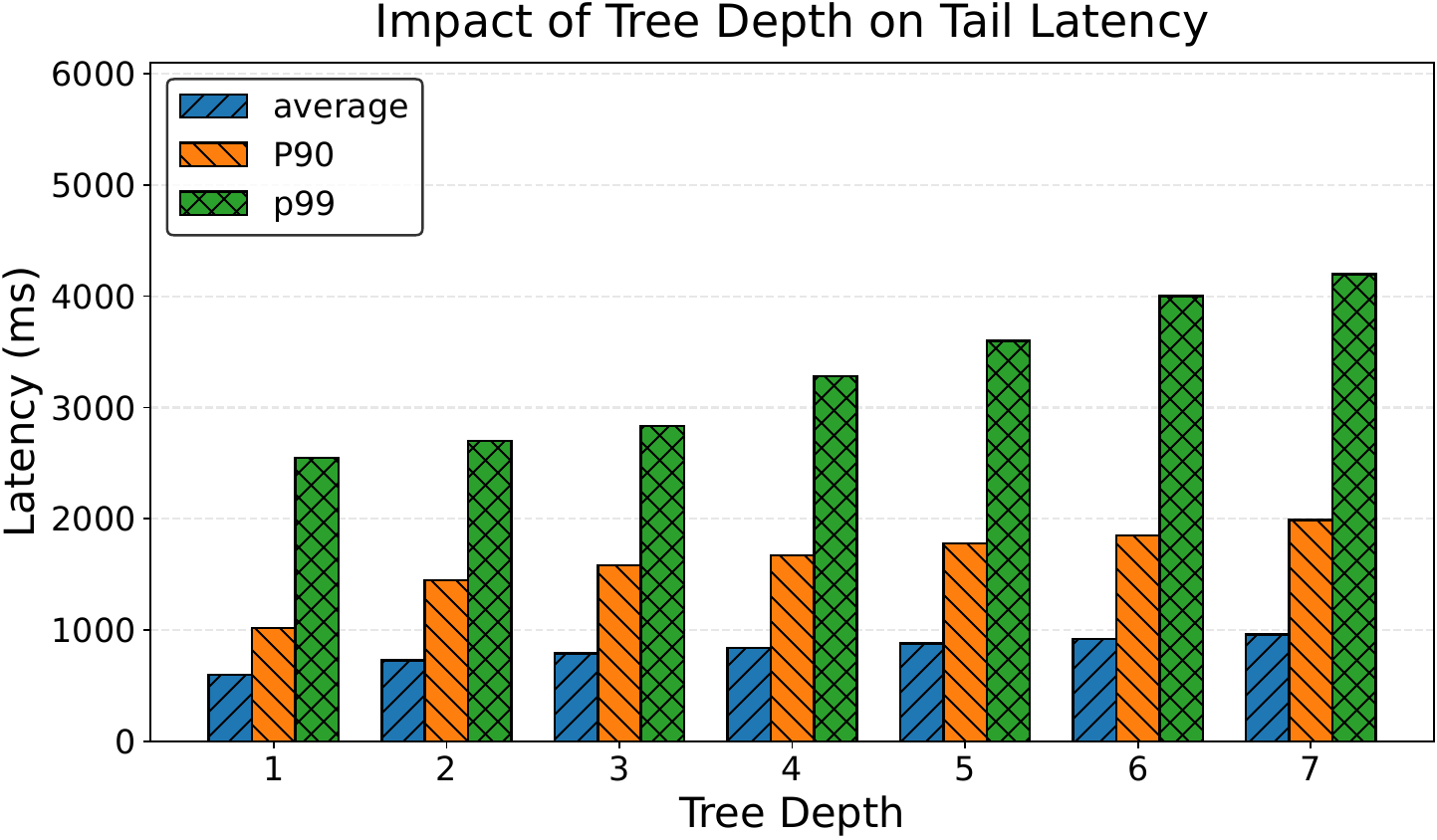}
    \caption{Latency vs.\ tree depth.}
    \label{fig:exp_depth_vs_latency}
  \end{subfigure}
  \caption{Tree depth sensitivity under transfer.}
  \label{fig:exp_depth_results}
\end{figure*}

We next increase the tree depth to directly stress the length of the coordination path.
Figure~\ref{fig:exp_depth_vs_tps} shows that throughput decreases monotonically with depth.
Unlike width, depth increases the number of coordination hops and the amount of synchronous logging that lies on the critical path, thereby reducing transaction completion rate under fixed client concurrency.
Notably, the throughput degradation is smooth and predictable, and we do not observe abrupt collapse,
indicating that the transfer-aware commit protocol remains stable even under artificially deep transfer chains.

Figure~\ref{fig:exp_depth_vs_latency} further confirms that latency grows with depth across all reported
percentiles, with P99 increasing more noticeably than mean/P90.
This trend directly matches the theoretical model in Section~\ref{sec:tree_2pc:summary}, where transaction response latency
requires $2H$ message round-trips and lock release latency requires $3H$ round-trips.
As $H$ grows, delays at each level (e.g., RPC scheduling, queuing, and replicated log synchronization)
accumulate along the critical path, amplifying tail latency.
Despite this increase, the growth remains proportional to depth rather than exhibiting runaway amplification,
supporting our earlier claim that transfer overhead is confined to localized commit paths rather than expanding
the global coordination scope.

These results validate the structural performance properties of our tree-shaped 2PC protocol under transfer.
Tree \emph{depth} is the dominant factor that increases coordination cost and tail latency, while tree \emph{width}
has limited impact on throughput and average latency and only moderately affects P99 through straggler probability.
The predictable sensitivity to commit-tree geometry demonstrates that our transfer-aware design introduces bounded,
localized overhead and maintains stable transaction performance under dynamic partition movement.

\subsection{Coordination Scalability: From Partition to Log-Stream Granularity}
\label{sec:exp:coord_scalability}

This subsection evaluates the fundamental coordination scalability difference between OceanBase 3.x and 4.x. While Section~7.5 studies the structural sensitivity of tree-shaped 2PC under transfer, here we examine a more fundamental architectural question: how the choice of coordination granularity affects the scalability of transaction commit.

OceanBase 3.x organizes 2PC participants at \emph{partition granularity}. Each partition independently participates in the commit protocol,
and a distributed transaction spanning $N$ partitions introduces $N$ participants.
In contrast, OceanBase 4.x introduces the \emph{log-stream granularity}, where multiple partitions are aggregated into a single log stream and participate in 2PC as one coordination unit.
This design reduces the effective participant cardinality and
redefines the coordination boundary of distributed transactions.

\subsubsection{Motivation}

In large-scale distributed systems, coordination scalability is a critical bottleneck.
Under partition granularity, the number of 2PC participants grows linearly with
the number of partitions involved in a transaction.
As the partition count increases, the system experiences increasing
message fan-out, acknowledgment aggregation overhead, scheduling contention,
and replicated log synchronization pressure.
Even if individual message handling is efficient,
the coordination cardinality itself becomes the dominant cost factor.

Log-stream granularity fundamentally changes this scaling behavior.
By aggregating partitions into log streams,
multiple data shards share a single coordination participant.
Thus, the number of 2PC participants grows with the number of log streams,
rather than the raw partition count.
Since log streams are significantly fewer than partitions,
this aggregation reduces coordination complexity from a partition-level scope
to a higher-level structural abstraction.

This experiment aims to quantitatively validate that the scalability gain
of 4.x stems from structural coordination reduction rather than incremental optimization.

\subsubsection{Experimental Setup}

We construct transactions that span varying numbers of partitions,
ranging from hundreds to tens of thousands.
For each configuration, we measure the 2PC commit latency
under both 3.x (partition granularity) and 4.x (log-stream granularity).
To isolate coordination overhead,
we measure only the commit phase latency while keeping workload intensity,
data distribution, and hardware configuration identical across versions.
No transfer operations are involved in this experiment,
ensuring that the measured cost reflects pure 2PC coordination behavior.

\subsubsection{Results and Analysis}

\begin{figure}[t]
  \centering
  \includegraphics[width=\linewidth]{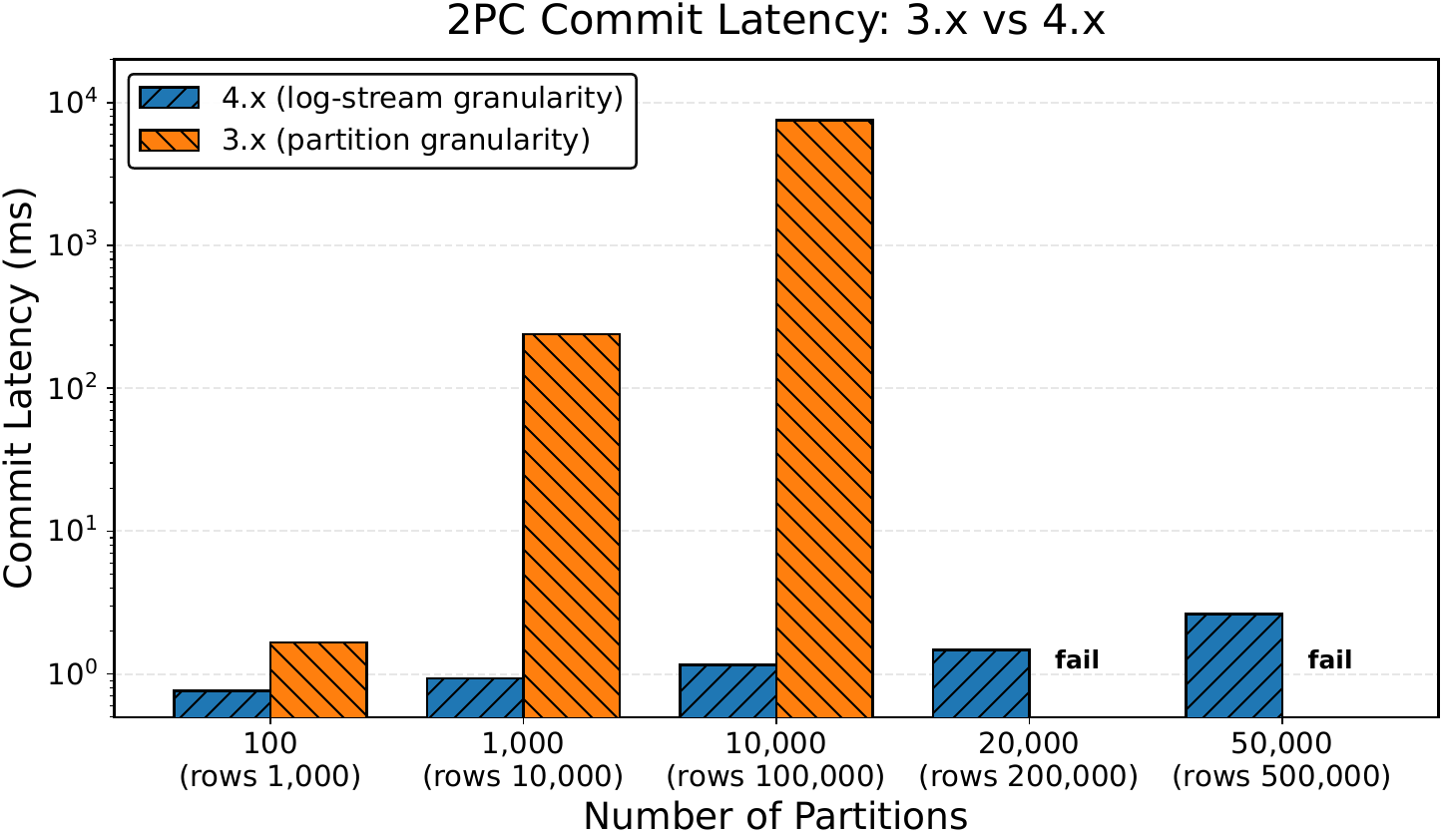}
  \caption{2PC commit latency comparison.}
  \label{fig:exp_commit_latency_compare}
\end{figure}

Figure~\ref{fig:exp_commit_latency_compare} presents the measured commit latency
as the number of partitions increases.
Under 3.x, commit latency grows rapidly with partition count.
The growth trend is close to linear, reflecting the fact that each additional
partition introduces another independent participant in the coordination phase.
As $N$ becomes large, the cumulative coordination overhead
significantly increases message aggregation delay, log synchronization pressure,
and scheduling contention.
At high partition counts, 3.x fails to complete commits within reasonable bounds,
demonstrating coordination amplification and scalability breakdown.

In contrast, 4.x exhibits substantially slower latency growth.
Although commit latency still increases with partition count,
the slope is significantly reduced.
This behavior confirms that log-stream aggregation effectively decouples
coordination cardinality from raw partition count.
Multiple partitions within the same log stream are represented by a single
participant in the commit tree, preventing participant explosion
and limiting the growth of coordination rounds.

Importantly, the scalability improvement is structural.
Rather than optimizing RPC processing or reducing logging cost,
4.x reduces the coordination domain itself.
This transformation changes the scaling characteristic of 2PC from
partition-level linear growth to log-stream-level bounded growth.
As a result, even at partition counts where 3.x becomes unstable,
4.x continues to operate predictably and maintain acceptable commit latency.

The results demonstrate that the transition from partition granularity to log-stream granularity is not merely an implementation refinement, but a fundamental architectural evolution. By redefining the atomic unit of migration and coordination, 4.x prevents coordination explosion and establishes a scalable foundation for transfer-aware distributed transactions. This structural improvement also explains why the tree-shaped 2PC design in 4.x maintains stable behavior under dynamic partition movement, as validated in Section~7.5.

Overall, this experiment confirms that log-stream aggregation significantly improves coordination scalability. Compared with partition-level 2PC in 3.x,
log-stream-level coordination in 4.x achieves dramatically better commit latency scaling as partition count increases, validating the architectural advancement of the proposed design.

\section{Related Work}
\label{sec:related_works}

Distributed transaction coordination has been a central challenge in large-scale database systems. Traditional 2PC protocols~\cite{lampson1993new, atif2009analysis, desai1996performance, samaras1995two} form the foundation of ACID guarantees but face scalability and latency issues in modern systems due to their reliance on per-partition coordination. Recent advancements in distributed databases address these limitations through novel optimizations and architectural innovations. This section reviews related work in three key areas: global transaction coordination, log-structured storage optimization, and dynamic partitioning with transfer handling.

\subsection{Global Transaction Coordination} 

Systems such as Spanner~\cite{corbett2013spanner} introduce hybrid logical clocks and TrueTime APIs to manage global transactions with strong consistency, but they still rely on per-partition coordination and thus inherit the scalability bottlenecks of traditional 2PC. Calvin~\cite{thomson2012calvin} shifts to deterministic execution and log-based scheduling to reduce coordination overhead, while CockroachDB~\cite{taft2020cockroachdb} uses hybrid logical clocks to optimize 2PC for geo-distributed environments. 

However, these approaches do not explicitly address the overhead caused by dynamic partition transfers, nor do they leverage log stream aggregation to reduce participant counts. Our framework differs by treating log streams as atomic participants, dramatically reducing coordination messages when multiple partitions reside on and share the same log stream.

\subsection{Log-Structured Storage and Optimization} 

Log-structured designs have gained traction for improving throughput and durability. Log-structured systems leverage log-structured storage for high-throughput transactional key-value stores, emphasizing atomicity and isolation. FASTER~\cite{chandramouli2018faster} further optimizes log-structured storage for low-latency operations, thereby demonstrating the potential of log-centric architectures. 

MaLT~\cite{malt_sigmod2025} manages large transaction state and recovery in OceanBase's LSM-tree storage engine. While these works focus on storage efficiency and throughput optimization, they do not integrate log streams as atomic units for transaction coordination—a capability that is central to our single-log-stream 2PC framework. Our approach extends log-structured principles to transaction coordination by using the log stream itself as the coordination unit rather than individual partitions.

\subsection{Dynamic Partitioning and Transfer} 

Existing solutions for dynamic partitioning, such as those in CockroachDB~\cite{taft2020cockroachdb}, often involve re-validation of participant lists post-transfer or introduce circular dependencies between partitions and logs~\cite{kolltveit2007circular, samaras1995two}. They require explicit participant list updates and complex recovery logic when partitions migrate during active transactions. In contrast, our tree-structured 2PC protocol inherently accommodates partition transfers by recursively tracking log stream dependencies, eliminating the need for explicit participant list updates and simplifying recovery logic. This design also addresses the ``circular transaction'' problem~\cite{kolltveit2007circular}, ensuring linearizable commit progression even during live migrations through automatic tree construction. When transactions contend on locks during transfer and commit, OceanBase relies on distributed deadlock detection and resolution (LCL~\cite{lcl_icde2023}, LCL+~\cite{lcl_plus_tocs2025}) to avoid indefinite blocking.

The main difference between the above related work and our tree-structured 2PC framework is that we take the \emph{log stream} as the atomic participant and coordination unit instead of the partition, and we natively support dynamic partition transfer and context loss (via recursive log stream trees and unknown-state mechanisms) without post-transfer re-validation of participant lists or circular dependency handling. While prior systems have made strides in distributed transaction management, they either (1) treat individual partitions as atomic participants, leading to scalability bottlenecks, or (2) struggle with dynamic topologies during transfers. Our work bridges these gaps by abstracting log streams as atomic units, reducing participant counts, and introducing an unknown-state mechanism to preserve consistency during context loss. These innovations align with the principles of modern log-structured systems~\cite{chandramouli2018faster} while addressing specific challenges in dynamic, cloud-native environments.

\section{Conclusions}
\label{sec:Conclusions}

This paper proposes a novel tree-structured 2PC framework for OceanBase that addresses scalability and latency challenges in distributed transaction coordination. By redefining the atomic unit from individual partitions to single-machine log streams, our approach reduces participant counts by orders of magnitude (e.g., 99\% reduction for 100 partitions), minimizing coordination overhead and bandwidth consumption. The tree-structured 2PC protocol dynamically handles partition transfers via a coordinator-rooted DAG topology, eliminating explicit participant list updates and resolving circular dependencies. During live migrations, transfer contexts are recursively embedded as leaf nodes in commit trees, ensuring linearizable progress under topology changes. Additionally, unknown-state mechanisms (prepare\_unknown, trans\_unknown) prevent consistency violations during context loss, avoiding erroneous aborts caused by misleading participant responses while isolating users from ambiguity.

Experimental evaluation on OceanBase 4.4.1 demonstrates that our framework achieves performance approaching that of single-machine transactions, with significantly reduced latency and bandwidth consumption. The system exhibits linear scalability under heavy workloads and maintains stable performance during concurrent partition transfers. This work bridges strong consistency with high performance, offering a scalable blueprint for next-generation distributed systems. By abstracting log streams as atomic units and introducing robust topology-handling mechanisms, our solution sets a new standard for low-latency, high-throughput transactions across geo-distributed architectures. The practical implementation in OceanBase validates its effectiveness in real-world scenarios, particularly for applications requiring elastic resource efficiency and dynamic partition management in cloud-native environments.



\bibliographystyle{ACM-Reference-Format}
\bibliography{sample.bib}

\newpage
\appendix
\section{Detailed Proofs}
\label{sec:appendix:detailed_proofs}

\subsection{Proof of Theorem~\ref{thm:log_stream_tree}}
\label{sec:appendix:log_stream_tree}
\begin{proof}
    Let $P$ be the set of partitions participating in concurrency control, and $L_{final}(p)$ denote the final log stream location of partition $p \in P$. The minimum set requirement is $\{L_{final}(p) : p \in P\}$.
    
    The log stream tree is constructed by including: (1) initial log streams where partitions in $P$ perform operations, and (2) all destination log streams from partition transfers, recursively following transfer records.
    
    \textbf{Completeness}: For any $p \in P$, if $p$ never transfers, $L_{final}(p)$ is the initial log stream (included by construction). If $p$ transfers along $L_0 \rightarrow \cdots \rightarrow L_k = L_{final}(p)$, recursive construction includes all $L_i$, hence $L_{final}(p)$.
    
    \textbf{Minimality}: A log stream is added only if it hosts partitions in $P$ (initially or via transfer), so every tree node corresponds to a valid partition location.
    
    Leaf nodes are exactly $\{L_{final}(p) : p \in P\}$. Since the minimum set requirement accepts any participant list containing this set, the tree satisfies it.
\end{proof}

\subsection{Proof of Theorem~\ref{thm:transfer_principle}}
\label{sec:appendix:transfer_principle}
\begin{proof}
We prove the transfer principle (Theorem~\ref{thm:transfer_principle}) by contradiction. The log formulation and the message formulation are equivalent because 2PC messages are generated from and reflect the logged state; we give the argument in terms of logs; it applies to messages in the same way.

\textbf{Part 1: Transfer-before logs must be migrated.}
Assume, for contradiction, that a transaction log $L$ written before the transfer-out log $T_{out}$ is not migrated to the destination log stream $dst$. By the minimum set requirement (Section~\ref{sec:background:ACID}), the participant list must include all log streams where partitions participating in concurrency control are finally located. Since the partition is transferred to $dst$, $dst$ becomes the final location. If $L$ is not migrated, then:
\begin{itemize}
    \item The destination log stream $dst$ lacks the transaction context from $L$, violating the completeness requirement for the participant list.
    \item The transaction state at $dst$ is incomplete, potentially causing incorrect commit/abort decisions.
    \item This violates atomicity, as the transaction's state is split between source and destination log streams.
\end{itemize}
Therefore, all logs written before $T_{out}$ must be migrated to $dst$ through the transfer process. The same requirement holds for 2PC messages sent before transfer completion, since those messages correspond to state that is (or will be) persisted in logs and must be applied to the destination via the transfer process.

\textbf{Part 2: Transfer-after logs must account for transfer impact.}
Assume, for contradiction, that a transaction log $L'$ written after $T_{out}$ does not account for the transfer's impact. This means $L'$ does not include the destination log stream $dst$ in its participant list or transaction context. Since the partition has been transferred to $dst$ (as recorded by $T_{out}$), $dst$ is now the final location of the partition. If $L'$ ignores this:
\begin{itemize}
    \item The participant list in $L'$ is incomplete, missing $dst$ which hosts the partition at commit time.
    \item This violates the minimum set requirement, as the participant list does not contain all final log stream locations.
    \item The 2PC protocol may fail to coordinate with $dst$, leading to inconsistent commit/abort decisions across log streams.
\end{itemize}
Therefore, all logs written after $T_{out}$ must include $dst$ in their participant lists (via \textit{interm\_parts} or \textit{incr\_parts}) to account for the transfer's impact. The same requirement holds for 2PC messages sent after transfer completion: they must be directed to (or include) $dst$ because messages are generated from the logged state and the participant list in that state must include $dst$.

The correctness of this approach is ensured by Algorithms~\ref{alg:partition_transfer} and \ref{alg:log_commit}, which guarantee that: (1) transfer migrates all pre-transfer transaction context (logs and the state reflected in messages), and (2) log commits atomically merge \textit{interm\_parts} into \textit{incr\_parts} before committing, ensuring complete participant lists for both logs and subsequent messages.
\end{proof}
    
\subsection{Proof of Theorem~\ref{thm:transfer_correctness}}
\label{sec:appendix:transfer_correctness}

\begin{proof}
    We show that the two requirements stated in Section~\ref{sec:transfer_con_transfer_commit} are satisfied, and that they imply Theorem~\ref{thm:transfer_principle}.
    
    \textbf{Requirement 1: Partition transfer must migrate all transaction context written before the transfer-out log to the destination.}
    
    Algorithm~\ref{alg:partition_transfer} acquires the transfer lock (line 1), then collects transaction context and blocks log writing for affected transactions (lines 2--9), then commits the transfer-out logs on source and destination (lines 10--11), and then adds $\mathit{dst}$ to $\mathit{interm\_parts}$ for affected transactions and unblocks them (lines 12--20). The transfer mechanism (of which this algorithm is the coordination part) is designed so that all transaction context written before the transfer-out log is migrated to the destination before or as part of the transfer that leads to committing that log; the commit of the transfer-out logs (lines 10--11) is the linearization point of ``transfer completion.'' Hence any transaction log or 2PC message sent before that commit is part of the pre-transfer context and is applied to the destination through the transfer process. Thus Requirement~1 is satisfied by the transfer process, and Algorithm~\ref{alg:partition_transfer} does not commit the transfer-out log before the migration of pre-transfer context is ensured.
    
    \textbf{Requirement 2: Transactions must include all previous transfer destination log streams in \textit{incr\_parts}.}
    
    We use the shared transfer lock and the atomicity of the two algorithms. Consider any transaction $t$ and any partition transfer that commits a transfer-out log for tablet $tab$ to destination $dst$ with timestamp $ts$, where $tab$ is involved in $t$ and $ts > t.ts'$ (so the transfer completes after $t$'s last relevant timestamp). When the 2PC log for $t$ is committed (via Algorithm~\ref{alg:log_commit}), we must have $dst \in t.\textit{incr\_parts}$ (because the committed log includes \textit{participants} and \textit{incr\_parts}, and $dst$ must be in the participant list per Theorem~\ref{thm:transfer_principle}).
    
    \textit{Case A:} The transfer commits its transfer-out logs (Algorithm~\ref{alg:partition_transfer}, lines 10--11) \emph{before} $t$ runs the merge and commit in Algorithm~\ref{alg:log_commit}. Then Algorithm~\ref{alg:partition_transfer} (lines 12--20) runs the loop over $tc$ and adds $dst$ to $t.\textit{interm\_parts}$ (since $tab$ is involved in $t$ and $ts > t.ts'$). When $t$ later runs Algorithm~\ref{alg:log_commit}, it merges \textit{interm\_parts} into \textit{incr\_parts} and commits (lines 6--7), so $dst \in t.\textit{incr\_parts}$ when the 2PC log is committed.
    
    \textit{Case B:} The transfer has not yet committed the transfer-out log when $t$ enters Algorithm~\ref{alg:log_commit}. Then the procedure acquires the lock (line 1), and for $t$ (if not stopped) merges \textit{interm\_parts} into \textit{incr\_parts} and commits the 2PC log (lines 4--7). When the transfer later runs Algorithm~\ref{alg:partition_transfer}, it commits the transfer-out log and then adds $dst$ to $t.\textit{interm\_parts}$. So this transfer completed \emph{after} $t$'s 2PC log commit; the participant list in that log need not include $dst$ (Theorem~\ref{thm:transfer_principle} requires logs written \emph{after} the transfer-out to account for the transfer). So no violation.
    
    \textit{Case C:} The transfer commits the transfer-out log \emph{after} $t$ acquires the lock in Algorithm~\ref{alg:log_commit} but \emph{before} $t$ merges and commits. Because the transfer must acquire the same lock (Algorithm~\ref{alg:partition_transfer}, line 1), the transfer waits until the procedure releases the lock (Algorithm~\ref{alg:log_commit}, line 10). So the transfer cannot commit in between $t$'s lock acquisition and release. Thus the only interleavings are Case A or B, and Requirement~2 holds.
    
    By the shared lock, exactly one of Case A or B applies for each (transfer, transaction) pair: either the transfer adds $dst$ to $t.\textit{interm\_parts}$ before $t$'s merge, or the transfer completes after $t$'s 2PC log commit. Hence every 2PC log commit by $t$ includes in $t.\textit{incr\_parts}$ every destination $dst$ of a transfer that completed before that commit and that involves a tablet in $t$.
    
\textbf{Satisfaction of Theorem~\ref{thm:transfer_principle}.}

    The transfer principle requires (1) logs and 2PC messages produced before the transfer-out log to be applied to the destination via the transfer process, and (2) logs and messages produced after the transfer-out log to include the destination (e.g., in \textit{incr\_parts}). (1) is guaranteed by Requirement~1. (2) is guaranteed by Requirement~2 and the fact that committed 2PC logs contain \textit{participants} and \textit{incr\_parts}, and 2PC messages are generated from that state. Thus Theorem~\ref{thm:transfer_principle} is satisfied.
\end{proof}

\subsection{Proof of Theorem~\ref{thm:safety_correctness}}
\label{sec:appendix:safety}

\begin{proof}
In the traditional flat 2PC model, it is well known that all participants remain in the same outcome (all committed or all aborted), and once a participant enters commit or abort it stays in that state. This safety property follows from the structure of the two-phase commit rules: commit is enabled only when all participants vote positively, while any negative vote forces a global abort.

The TLA+ specification in Appendix~\ref{sec:appendix:tla_spec} models the same decision logic. A node may enter the commit state only after collecting positive votes from all its children or after receiving a commit decision from its parent. Similarly, a node may enter the abort state only if a negative vote is observed or an abort decision is received from its parent. Thus, conflicting local decisions cannot be generated within a subtree.

The tree-shaped protocol is a hierarchical implementation of the same 2PC logic. Each internal node first coordinates its subtree to reach a decision and then behaves as a participant toward its parent. This is equivalent to collapsing each decided subtree into a single participant in a flat 2PC execution. Therefore, every execution of the tree-shaped protocol corresponds to a valid execution of the traditional flat 2PC specification.

Since the flat 2PC protocol satisfies the agreement invariant (no partition of participants into committed and aborted states), and the tree-shaped protocol follows the same state transitions and message rules, the safety property of the tree-shaped 2PC protocol follows directly.
\end{proof}

\subsection{Proof of Theorem~\ref{thm:liveness_correctness}}
\label{sec:appendix:liveness}

\begin{proof}
The liveness property of the traditional 2PC protocol states that, under reliable message delivery and fairness of enabled actions, every participant eventually reaches a terminal state and all participants agree on the same outcome.

The specification in Appendix~\ref{sec:appendix:tla_spec} preserves these assumptions. Messages are modeled as reliably delivered, and enabled protocol actions are assumed to eventually execute. When the root initiates the protocol, prepare requests are propagated along the tree. Each node eventually collects votes from its children, after which either a commit or abort decision becomes enabled. The decision is then propagated downward, and acknowledgements are returned upward. Since the system is finite and fairness ensures that no enabled action remains forever postponed, every node eventually reaches a terminal decision and completes the protocol.

As in the safety argument, each subtree execution can be viewed as a participant in an equivalent flat 2PC execution. Therefore, every fair execution of the tree-shaped protocol corresponds to a fair execution of the traditional 2PC specification, and termination follows directly.
\end{proof}

\section{State Machine}
\label{sec:appendix:state_machine}

\textit{Note:} This appendix specifies the core state machine
consistent with the TLA+ specification in Appendix~\ref{sec:appendix:tla_spec}.
Optimizations described in the main text are applied on top of this specification.

\noindent$\bullet\ $ handle\_2pc\_prepare\_request($p$)

\begin{itemize}
    \item Enable: When the state is RUNNING and a PrepareReq request is received from a parent participant $p$, or when a transaction commit request is issued at the root.
    \item Action: If $p$ exists and the parent is not yet recorded, record $p$ as the parent; merge pending participants into the current participant list; enter the PREPARE state; initialize votes of all participants as unknown; send PrepareReq messages to all participants in the participant list.
\end{itemize}

\noindent$\bullet\ $ handle\_duplicate\_prepare\_request($p$)

\begin{itemize}
    \item Enable: When the state is PREPARE and a PrepareReq request is received from a participant that is not the recorded parent.
    \item Action: Reply PrepareResp with status OK without changing the current state.
\end{itemize}

\noindent$\bullet\ $ handle\_orphan\_prepare\_request($p$)

\begin{itemize}
    \item Enable: When the state is ABORT or TOMBSTONE and a PrepareReq request is received.
    \item Action: Reply PrepareResp with status NO.
\end{itemize}

\noindent$\bullet\ $ handle\_2pc\_prepare\_response($p$, $p\_status$)

\begin{itemize}
    \item Enable: When the state is PREPARE and a PrepareResp message with $p\_status$ is received from a child participant $p$.
    \item Action: Update the vote of participant $p$ to OK or NO accordingly.
\end{itemize}

\noindent$\bullet\ $ handle\_2pc\_commit\_decided()

\begin{itemize}
    \item Enable: When the state is PREPARE and all collected votes indicate success.
    \item Action: Merge pending participants; enter the COMMIT state; initialize acknowledgement tracking; if the parent participant exists, send PrepareResp with status OK to the parent; otherwise, send Commit messages to all participants in the participant list.
\end{itemize}

\noindent$\bullet\ $ handle\_2pc\_abort\_decided()

\begin{itemize}
    \item Enable: When the state is PREPARE and at least one vote indicates failure.
    \item Action: Merge pending participants; enter the ABORT state; initialize acknowledgement tracking; if the parent participant exists, send PrepareResp with status NO to the parent; otherwise, send Abort messages to all participants in the participant list.
\end{itemize}

\noindent$\bullet\ $ handle\_2pc\_commit\_request($p$)

\begin{itemize}
    \item Enable: When the state is RUNNING or PREPARE and a Commit message is received from the parent participant $p$.
    \item Action: Merge pending participants; enter the COMMIT state; send Commit messages to all participants; send Ack to the parent participant.
\end{itemize}

\noindent$\bullet\ $ handle\_2pc\_abort\_request($p$)

\begin{itemize}
    \item Enable: When the state is RUNNING or PREPARE and an Abort message is received from the parent participant $p$.
    \item Action: Merge pending participants; enter the ABORT state; send Abort messages to all participants; send Ack to the parent participant.
\end{itemize}

\noindent$\bullet\ $ handle\_orphan\_commit\_request($p$)

\begin{itemize}
    \item Enable: When the state is COMMIT or TOMBSTONE and a Commit message is received.
    \item Action: Reply Ack without changing the local state.
\end{itemize}

\noindent$\bullet\ $ handle\_orphan\_abort\_request($p$)

\begin{itemize}
    \item Enable: When the state is ABORT or TOMBSTONE and an Abort message is received.
    \item Action: Reply Ack without changing the local state.
\end{itemize}

\noindent$\bullet\ $ handle\_2pc\_ack\_response($p$)

\begin{itemize}
    \item Enable: When the state is COMMIT or ABORT and an Ack message is received from the child participant $p$.
    \item Action: Mark the acknowledgement of participant $p$ as received.
\end{itemize}

\noindent$\bullet\ $ forget\_ctx()

\begin{itemize}
    \item Enable: When the state is COMMIT or ABORT and acknowledgements from all participants are received.
    \item Action: Enter the TOMBSTONE state.
\end{itemize}

\noindent$\bullet\ $ add\_intermediate\_participant($p$)

\begin{itemize}
    \item Enable: When the state is RUNNING, PREPARE, COMMIT, or ABORT and a new participant $p$ is introduced by transfer.
    \item Action: Add $p$ into the pending participant list; it will be merged at the next phase transition.
\end{itemize}

\section{Full TLA+ Specification}
\label{sec:appendix:tla_spec}
\textit{Note:} This appendix gives the full TLA+ specification (module \texttt{NewTwoPhaseCommit}) for the flat 2PC model, with TM/RM roles and message types; it is consistent with the state machine in Appendix~\ref{sec:appendix:state_machine}.

\input{tla+/tla+}

\end{document}

%% file: tla+/tla+.tex

\tlatex
\setboolean{shading}{true}

\@x{}\moduleLeftDash\@xx{ {\MODULE} 2pc\_tla}\moduleRightDash\@xx{}%
\@x{ {\EXTENDS} Naturals ,\, Sequences ,\, FiniteSets ,\, TLC}%
\@pvspace{8.0pt}%

\@x{ {\CONSTANTS}}%
\@x{\@s{16.4} Node ,\,}%
\@x{\@s{16.4} Root ,\,}%
\@x{\@s{16.4} InitChildren}%
\@pvspace{8.0pt}%
\@x{ {\VARIABLES}}%
\@x{\@s{16.4} rmState ,\,}%
\@x{\@s{16.4} children ,\,}%
\@x{\@s{16.4} intermediate\_children ,\,}%
\@x{\@s{16.4} msgs ,\,}%
\@x{\@s{16.4} votes ,\,}%
\@x{\@s{16.4} acks ,\,}%
\@x{\@s{16.4} parent}%
\@pvspace{8.0pt}%
\@x{ Vars \.{\defeq} {\langle} rmState ,\, children ,\,}%
\@x{\@s{41.61} intermediate\_children ,\,}%
\@x{\@s{41.61} msgs ,\, votes ,\, acks ,\, parent {\rangle}}%
\@pvspace{8.0pt}%

\@x{ States \.{\defeq}}%
\@x{\@s{16.4} \{\@w{RUNNING} ,\,\@w{PREPARE} ,\,\@w{COMMIT} ,\,\@w{ABORT} ,\,\@w{TOMBSTONE} \}}%
\@pvspace{8.0pt}%
\@x{ MsgPrepareReq ( src ,\, dst ) \.{\defeq}}%
\@x{\@s{16.4} [ type \.{\mapsto}\@w{PrepareReq} ,\,}%
\@x{\@s{16.4} src \.{\mapsto} src ,\, dst \.{\mapsto} dst ]}%
\@pvspace{8.0pt}%
\@x{ MsgPrepareResp ( src ,\, dst ,\, status ) \.{\defeq}}%
\@x{\@s{16.4} [ type \.{\mapsto}\@w{PrepareResp} ,\, src \.{\mapsto} src ,\,}%
\@x{\@s{16.4} dst \.{\mapsto} dst ,\, status \.{\mapsto} status ]}%
\@pvspace{8.0pt}%
\@x{ MsgCommit ( src ,\, dst ) \.{\defeq}}%
\@x{\@s{16.4} [ type \.{\mapsto}\@w{Commit} ,\,}%
\@x{\@s{16.4} src \.{\mapsto} src ,\, dst \.{\mapsto} dst ]}%
\@pvspace{8.0pt}%
\@x{ MsgAbort ( src ,\, dst ) \.{\defeq}}%
\@x{\@s{16.4} [ type \.{\mapsto}\@w{Abort} ,\,}%
\@x{\@s{16.4} src \.{\mapsto} src ,\, dst \.{\mapsto} dst ]}%
\@pvspace{8.0pt}%
\@x{ MsgAck ( src ,\, dst ) \.{\defeq}}%
\@x{\@s{16.4} [ type \.{\mapsto}\@w{Ack} ,\,}%
\@x{\@s{16.4} src \.{\mapsto} src ,\, dst \.{\mapsto} dst ]}%
\@pvspace{8.0pt}%
\@x{ IsRoot ( n ) \.{\defeq} n \.{=} Root}%
\@pvspace{8.0pt}%
\@x{ MergedChildren ( n ) \.{\defeq}}%
\@x{\@s{16.4} children [ n ] \.{\cup} intermediate\_children [ n ]}%
\@pvspace{8.0pt}%
\@x{ ApplyMerge ( n ,\, mc ) \.{\defeq}}%
\@x{\@s{16.4} \.{\land} children \.{'} \.{=}}%
\@x{\@s{27.51} [ children {\EXCEPT} {\bang} [ n ] \.{=} mc ]}%
\@x{\@s{16.4} \.{\land} intermediate\_children \.{'} \.{=}}%
\@x{\@s{27.51} [ intermediate\_children {\EXCEPT} {\bang} [ n ] \.{=} \{ \} ]}%
\@pvspace{8.0pt}%
\@x{ RecordParent ( n ,\, src ) \.{\defeq}}%
\@x{\@s{16.4} [ parent {\EXCEPT} {\bang} [ n ] \.{=}}%
\@x{\@s{27.51} {\IF} @ \.{=}\@w{none} \.{\THEN} src \.{\ELSE} @ ]}%
\@pvspace{8.0pt}%
\@x{ AllVotesOk ( n ) \.{\defeq}}%
\@x{\@s{16.4} \A\, c \.{\in} children [ n ] \.{:} votes [ n ] [ c ] \.{=}\@w{ok}}%
\@pvspace{8.0pt}%
\@x{ AnyVoteNo ( n ) \.{\defeq}}%
\@x{\@s{16.4} \E\, c \.{\in} children [ n ] \.{:} votes [ n ] [ c ] \.{=}\@w{no}}%
\@pvspace{8.0pt}%
\@x{ AllAcked ( n ) \.{\defeq}}%
\@x{\@s{16.4} \A\, c \.{\in} children [ n ] \.{:} acks [ n ] [ c ] \.{=} {\TRUE}}%
\@pvspace{8.0pt}%

\@x{ Init \.{\defeq}}%
\@x{\@s{16.4} \.{\land} rmState \.{=} [ n \.{\in} Node \.{\mapsto}\@w{RUNNING} ]}%
\@x{\@s{16.4} \.{\land} children \.{=} InitChildren}%
\@x{\@s{16.4} \.{\land} intermediate\_children \.{=}}%
\@x{\@s{27.51} [ n \.{\in} Node \.{\mapsto} \{ \} ]}%
\@x{\@s{16.4} \.{\land} msgs \.{=} \{ \}}%
\@x{\@s{16.4} \.{\land} votes \.{=} [ n \.{\in} Node \.{\mapsto}}%
\@x{\@s{27.51} [ c \.{\in} children [ n ] \.{\mapsto}\@w{unknown} ] ]}%
\@x{\@s{16.4} \.{\land} acks \.{=} [ n \.{\in} Node \.{\mapsto}}%
\@x{\@s{27.51} [ c \.{\in} children [ n ] \.{\mapsto} {\FALSE} ] ]}%
\@x{\@s{16.4} \.{\land} parent \.{=} [ n \.{\in} Node \.{\mapsto}\@w{none} ]}%
\@pvspace{8.0pt}%

\@x{ RootStartToCommit \.{\defeq}}%
\@x{\@s{16.4} \.{\land} rmState [ Root ] \.{=}\@w{RUNNING}}%
\@x{\@s{16.4} \.{\land} {\LET} mc \.{\defeq} MergedChildren ( Root ) {\IN}}%
\@x{\@s{27.51} \.{\land} rmState \.{'} \.{=}}%
\@x{\@s{38.92} [ rmState {\EXCEPT} {\bang} [ Root ] \.{=}\@w{PREPARE} ]}%
\@x{\@s{27.51} \.{\land} ApplyMerge ( Root ,\, mc )}%
\@x{\@s{27.51} \.{\land} votes \.{'} \.{=} [ votes {\EXCEPT} {\bang} [ Root ] \.{=}}%
\@x{\@s{38.92} [ c \.{\in} mc \.{\mapsto}\@w{unknown} ] ]}%
\@x{\@s{27.51} \.{\land} msgs \.{'} \.{=} msgs \.{\cup}}%
\@x{\@s{38.92} \{ MsgPrepareReq ( Root ,\, c ) \.{:} c \.{\in} mc \}}%
\@x{\@s{27.51} \.{\land} {\UNCHANGED} {\langle} acks ,\, parent {\rangle}}%
\@pvspace{8.0pt}%
\@x{ Handle2pcPrepareRequest ( n ) \.{\defeq}}%
\@x{\@s{16.4} \.{\land} rmState [ n ] \.{=}\@w{RUNNING}}%
\@x{\@s{16.4} \.{\land} \E\, m \.{\in} msgs \.{:}}%
\@x{\@s{31.61} \.{\land} m \.{.} type \.{=}\@w{PrepareReq}}%
\@x{\@s{31.61} \.{\land} m \.{.} dst \.{=} n}%
\@x{\@s{31.61} \.{\land} {\LET} mc \.{\defeq} MergedChildren ( n ) {\IN}}%
\@x{\@s{43.02} \.{\land} parent \.{'} \.{=}}%
\@x{\@s{54.43} [ parent {\EXCEPT} {\bang} [ n ] \.{=} m \.{.} src ]}%
\@x{\@s{43.02} \.{\land} rmState \.{'} \.{=}}%
\@x{\@s{54.43} [ rmState {\EXCEPT} {\bang} [ n ] \.{=}\@w{PREPARE} ]}%
\@x{\@s{43.02} \.{\land} ApplyMerge ( n ,\, mc )}%
\@x{\@s{43.02} \.{\land} votes \.{'} \.{=} [ votes {\EXCEPT} {\bang} [ n ] \.{=}}%
\@x{\@s{54.43} [ c \.{\in} mc \.{\mapsto}\@w{unknown} ] ]}%
\@x{\@s{43.02} \.{\land} msgs \.{'} \.{=} msgs \.{\cup}}%
\@x{\@s{54.43} \{ MsgPrepareReq ( n ,\, c ) \.{:} c \.{\in} mc \}}%
\@x{\@s{43.02} \.{\land} {\UNCHANGED} {\langle} acks {\rangle}}%
\@pvspace{8.0pt}%
\@x{ Handle2pcDuplicatePrepareRequest ( n ) \.{\defeq}}%
\@x{\@s{16.4} \.{\land} rmState [ n ] \.{=}\@w{PREPARE}}%
\@x{\@s{16.4} \.{\land} \E\, m \.{\in} msgs \.{:}}%
\@x{\@s{31.61} \.{\land} m \.{.} type \.{=}\@w{PrepareReq}}%
\@x{\@s{31.61} \.{\land} m \.{.} dst \.{=} n}%
\@x{\@s{31.61} \.{\land} m \.{.} src \.{\neq} parent [ n ]}%
\@x{\@s{31.61} \.{\land} msgs \.{'} \.{=} msgs \.{\cup}}%
\@x{\@s{43.02} \{ MsgPrepareResp ( n ,\, m \.{.} src ,\,\@w{ok} ) \}}%
\@x{\@s{31.61} \.{\land} {\UNCHANGED} {\langle} rmState ,\, children ,\,}%
\@x{\@s{43.02} intermediate\_children ,\,}%
\@x{\@s{43.02} votes ,\, acks ,\, parent {\rangle}}%
\@pvspace{8.0pt}%
\@x{ HandleOrphan2pcPrepareRequest ( n ) \.{\defeq}}%
\@x{\@s{16.4} \.{\land} rmState [ n ] \.{\in} \{\@w{ABORT} ,\,\@w{TOMBSTONE} \}}%
\@x{\@s{16.4} \.{\land} \E\, m \.{\in} msgs \.{:}}%
\@x{\@s{31.61} \.{\land} m \.{.} type \.{=}\@w{PrepareReq}}%
\@x{\@s{31.61} \.{\land} m \.{.} dst \.{=} n}%
\@x{\@s{31.61} \.{\land} parent \.{'} \.{=} RecordParent ( n ,\, m \.{.} src )}%
\@x{\@s{31.61} \.{\land} msgs \.{'} \.{=} msgs \.{\cup}}%
\@x{\@s{43.02} \{ MsgPrepareResp ( n ,\, m \.{.} src ,\,\@w{no} ) \}}%
\@x{\@s{31.61} \.{\land} {\UNCHANGED} {\langle} rmState ,\, children ,\,}%
\@x{\@s{43.02} intermediate\_children ,\,}%
\@x{\@s{43.02} votes ,\, acks {\rangle}}%
\@pvspace{8.0pt}%

\@x{ Handle2pcPrepareResponse ( n ) \.{\defeq}}%
\@x{\@s{16.4} \.{\land} rmState [ n ] \.{=}\@w{PREPARE}}%
\@x{\@s{16.4} \.{\land} \E\, m \.{\in} msgs \.{:}}%
\@x{\@s{31.61} \.{\land} m \.{.} type \.{=}\@w{PrepareResp}}%
\@x{\@s{31.61} \.{\land} m \.{.} dst \.{=} n}%
\@x{\@s{31.61} \.{\land} m \.{.} src \.{\in} children [ n ]}%
\@x{\@s{31.61} \.{\land} votes [ n ] [ m \.{.} src ] \.{=}\@w{unknown}}%
\@x{\@s{31.61} \.{\land} votes \.{'} \.{=}}%
\@x{\@s{43.02} [ votes {\EXCEPT} {\bang} [ n ] [ m \.{.} src ] \.{=} m \.{.} status ]}%
\@x{\@s{31.61} \.{\land} {\UNCHANGED} {\langle} rmState ,\, children ,\,}%
\@x{\@s{43.02} intermediate\_children ,\,}%
\@x{\@s{43.02} msgs ,\, acks ,\, parent {\rangle}}%
\@pvspace{8.0pt}%

\@x{ Handle2pcCommitDecided ( n ) \.{\defeq}}%
\@x{\@s{16.4} \.{\land} rmState [ n ] \.{=}\@w{PREPARE}}%
\@x{\@s{16.4} \.{\land} AllVotesOk ( n )}%
\@x{\@s{16.4} \.{\land} {\IF} IsRoot ( n )}%
\@x{\@s{27.51} \.{\THEN} {\LET} mc \.{\defeq} MergedChildren ( n ) {\IN}}%
\@x{\@s{38.92} \.{\land} rmState \.{'} \.{=}}%
\@x{\@s{50.33} [ rmState {\EXCEPT} {\bang} [ n ] \.{=}\@w{COMMIT} ]}%
\@x{\@s{38.92} \.{\land} ApplyMerge ( n ,\, mc )}%
\@x{\@s{38.92} \.{\land} acks \.{'} \.{=} [ acks {\EXCEPT} {\bang} [ n ] \.{=}}%
\@x{\@s{50.33} [ c \.{\in} mc \.{\mapsto} {\FALSE} ] ]}%
\@x{\@s{38.92} \.{\land} msgs \.{'} \.{=} msgs \.{\cup}}%
\@x{\@s{50.33} \{ MsgCommit ( n ,\, c ) \.{:} c \.{\in} mc \}}%
\@x{\@s{38.92} \.{\land} {\UNCHANGED} {\langle} votes ,\, parent {\rangle}}%
\@x{\@s{27.51} \.{\ELSE} \.{\land} msgs \.{'} \.{=} msgs \.{\cup}}%
\@x{\@s{38.92} \{ MsgPrepareResp ( n ,\, parent [ n ] ,\,\@w{ok} ) \}}%
\@x{\@s{38.92} \.{\land} {\UNCHANGED} {\langle} rmState ,\, children ,\,}%
\@x{\@s{50.33} intermediate\_children ,\,}%
\@x{\@s{50.33} votes ,\, acks ,\, parent {\rangle}}%
\@pvspace{8.0pt}%
\@x{ Handle2pcAbortDecided ( n ) \.{\defeq}}%
\@x{\@s{16.4} \.{\land} \.{\lor} rmState [ n ] \.{=}\@w{PREPARE}}%
\@x{\@s{27.51} \.{\lor} ( IsRoot ( n ) \.{\land} rmState [ n ] \.{=}\@w{RUNNING} )}%
\@x{\@s{16.4} \.{\land} AnyVoteNo ( n )}%
\@x{\@s{16.4} \.{\land} {\IF} IsRoot ( n )}%
\@x{\@s{27.51} \.{\THEN} {\LET} mc \.{\defeq} MergedChildren ( n ) {\IN}}%
\@x{\@s{38.92} \.{\land} rmState \.{'} \.{=}}%
\@x{\@s{50.33} [ rmState {\EXCEPT} {\bang} [ n ] \.{=}\@w{ABORT} ]}%
\@x{\@s{38.92} \.{\land} ApplyMerge ( n ,\, mc )}%
\@x{\@s{38.92} \.{\land} acks \.{'} \.{=} [ acks {\EXCEPT} {\bang} [ n ] \.{=}}%
\@x{\@s{50.33} [ c \.{\in} mc \.{\mapsto} {\FALSE} ] ]}%
\@x{\@s{38.92} \.{\land} msgs \.{'} \.{=} msgs \.{\cup}}%
\@x{\@s{50.33} \{ MsgAbort ( n ,\, c ) \.{:} c \.{\in} mc \}}%
\@x{\@s{38.92} \.{\land} {\UNCHANGED} {\langle} votes ,\, parent {\rangle}}%
\@x{\@s{27.51} \.{\ELSE} \.{\land} msgs \.{'} \.{=} msgs \.{\cup}}%
\@x{\@s{38.92} \{ MsgPrepareResp ( n ,\, parent [ n ] ,\,\@w{no} ) \}}%
\@x{\@s{38.92} \.{\land} {\UNCHANGED} {\langle} rmState ,\, children ,\,}%
\@x{\@s{50.33} intermediate\_children ,\,}%
\@x{\@s{50.33} votes ,\, acks ,\, parent {\rangle}}%
\@pvspace{8.0pt}%

\@x{ Handle2pcCommitRequest ( n ) \.{\defeq}}%
\@x{\@s{16.4} \.{\land} {\lnot} IsRoot ( n )}%
\@x{\@s{16.4} \.{\land} rmState [ n ] \.{\in} \{\@w{RUNNING} ,\,\@w{PREPARE} \}}%
\@x{\@s{16.4} \.{\land} \E\, m \.{\in} msgs \.{:}}%
\@x{\@s{31.61} \.{\land} m \.{.} type \.{=}\@w{Commit}}%
\@x{\@s{31.61} \.{\land} m \.{.} dst \.{=} n}%
\@x{\@s{31.61} \.{\land} {\LET} mc \.{\defeq} MergedChildren ( n ) {\IN}}%
\@x{\@s{43.02} \.{\land} parent \.{'} \.{=} RecordParent ( n ,\, m \.{.} src )}%
\@x{\@s{43.02} \.{\land} rmState \.{'} \.{=}}%
\@x{\@s{54.43} [ rmState {\EXCEPT} {\bang} [ n ] \.{=}\@w{COMMIT} ]}%
\@x{\@s{43.02} \.{\land} ApplyMerge ( n ,\, mc )}%
\@x{\@s{43.02} \.{\land} acks \.{'} \.{=} [ acks {\EXCEPT} {\bang} [ n ] \.{=}}%
\@x{\@s{54.43} [ c \.{\in} mc \.{\mapsto} {\FALSE} ] ]}%
\@x{\@s{43.02} \.{\land} msgs \.{'} \.{=} msgs \.{\cup}}%
\@x{\@s{54.43} \{ MsgCommit ( n ,\, c ) \.{:} c \.{\in} mc \}}%
\@x{\@s{54.43} \.{\cup} \{ MsgAck ( n ,\, m \.{.} src ) \}}%
\@x{\@s{43.02} \.{\land} {\UNCHANGED} {\langle} votes {\rangle}}%
\@pvspace{8.0pt}%
\@x{ Handle2pcAbortRequest ( n ) \.{\defeq}}%
\@x{\@s{16.4} \.{\land} {\lnot} IsRoot ( n )}%
\@x{\@s{16.4} \.{\land} rmState [ n ] \.{\in} \{\@w{RUNNING} ,\,\@w{PREPARE} \}}%
\@x{\@s{16.4} \.{\land} \E\, m \.{\in} msgs \.{:}}%
\@x{\@s{31.61} \.{\land} m \.{.} type \.{=}\@w{Abort}}%
\@x{\@s{31.61} \.{\land} m \.{.} dst \.{=} n}%
\@x{\@s{31.61} \.{\land} {\LET} mc \.{\defeq} MergedChildren ( n ) {\IN}}%
\@x{\@s{43.02} \.{\land} parent \.{'} \.{=} RecordParent ( n ,\, m \.{.} src )}%
\@x{\@s{43.02} \.{\land} rmState \.{'} \.{=}}%
\@x{\@s{54.43} [ rmState {\EXCEPT} {\bang} [ n ] \.{=}\@w{ABORT} ]}%
\@x{\@s{43.02} \.{\land} ApplyMerge ( n ,\, mc )}%
\@x{\@s{43.02} \.{\land} acks \.{'} \.{=} [ acks {\EXCEPT} {\bang} [ n ] \.{=}}%
\@x{\@s{54.43} [ c \.{\in} mc \.{\mapsto} {\FALSE} ] ]}%
\@x{\@s{43.02} \.{\land} msgs \.{'} \.{=} msgs \.{\cup}}%
\@x{\@s{54.43} \{ MsgAbort ( n ,\, c ) \.{:} c \.{\in} mc \}}%
\@x{\@s{54.43} \.{\cup} \{ MsgAck ( n ,\, m \.{.} src ) \}}%
\@x{\@s{43.02} \.{\land} {\UNCHANGED} {\langle} votes {\rangle}}%
\@pvspace{8.0pt}%
\@x{ HandleOrphan2pcCommitRequest ( n ) \.{\defeq}}%
\@x{\@s{16.4} \.{\land} {\lnot} IsRoot ( n )}%
\@x{\@s{16.4} \.{\land} rmState [ n ] \.{\in} \{\@w{COMMIT} ,\,\@w{TOMBSTONE} \}}%
\@x{\@s{16.4} \.{\land} \E\, m \.{\in} msgs \.{:}}%
\@x{\@s{31.61} \.{\land} m \.{.} type \.{=}\@w{Commit}}%
\@x{\@s{31.61} \.{\land} m \.{.} dst \.{=} n}%
\@x{\@s{31.61} \.{\land} msgs \.{'} \.{=} msgs \.{\cup}}%
\@x{\@s{43.02} \{ MsgAck ( n ,\, m \.{.} src ) \}}%
\@x{\@s{31.61} \.{\land} {\UNCHANGED} {\langle} rmState ,\, children ,\,}%
\@x{\@s{43.02} intermediate\_children ,\,}%
\@x{\@s{43.02} votes ,\, acks ,\, parent {\rangle}}%
\@pvspace{8.0pt}%
\@x{ HandleOrphan2pcAbortRequest ( n ) \.{\defeq}}%
\@x{\@s{16.4} \.{\land} {\lnot} IsRoot ( n )}%
\@x{\@s{16.4} \.{\land} rmState [ n ] \.{\in} \{\@w{ABORT} ,\,\@w{TOMBSTONE} \}}%
\@x{\@s{16.4} \.{\land} \E\, m \.{\in} msgs \.{:}}%
\@x{\@s{31.61} \.{\land} m \.{.} type \.{=}\@w{Abort}}%
\@x{\@s{31.61} \.{\land} m \.{.} dst \.{=} n}%
\@x{\@s{31.61} \.{\land} msgs \.{'} \.{=} msgs \.{\cup}}%
\@x{\@s{43.02} \{ MsgAck ( n ,\, m \.{.} src ) \}}%
\@x{\@s{31.61} \.{\land} {\UNCHANGED} {\langle} rmState ,\, children ,\,}%
\@x{\@s{43.02} intermediate\_children ,\,}%
\@x{\@s{43.02} votes ,\, acks ,\, parent {\rangle}}%
\@pvspace{8.0pt}%

\@x{ InternalAbort ( n ) \.{\defeq}}%
\@x{\@s{16.4} \.{\land} rmState [ n ] \.{=}\@w{RUNNING}}%
\@x{\@s{16.4} \.{\land} {\LET} mc \.{\defeq} MergedChildren ( n )}%
\@x{\@s{27.51} parentNotify \.{\defeq}}%
\@x{\@s{38.92} {\IF} parent [ n ] \.{\neq}\@w{none}}%
\@x{\@s{38.92} \.{\THEN} \{ MsgPrepareResp ( n ,\, parent [ n ] ,\,\@w{no} ) \}}%
\@x{\@s{38.92} \.{\ELSE} \{ \}}%
\@x{\@s{16.4} {\IN}}%
\@x{\@s{27.51} \.{\land} rmState \.{'} \.{=}}%
\@x{\@s{38.92} [ rmState {\EXCEPT} {\bang} [ n ] \.{=}\@w{ABORT} ]}%
\@x{\@s{27.51} \.{\land} ApplyMerge ( n ,\, mc )}%
\@x{\@s{27.51} \.{\land} acks \.{'} \.{=} [ acks {\EXCEPT} {\bang} [ n ] \.{=}}%
\@x{\@s{38.92} [ c \.{\in} mc \.{\mapsto} {\FALSE} ] ]}%
\@x{\@s{27.51} \.{\land} msgs \.{'} \.{=} msgs \.{\cup}}%
\@x{\@s{38.92} \{ MsgAbort ( n ,\, c ) \.{:} c \.{\in} mc \}}%
\@x{\@s{38.92} \.{\cup} parentNotify}%
\@x{ \.{\land} {\UNCHANGED} {\langle} votes ,\, parent {\rangle}}%
\@pvspace{8.0pt}%

\@x{ Handle2pcAckResponse ( n ) \.{\defeq}}%
\@x{\@s{16.4} \.{\land} rmState [ n ] \.{\in} \{\@w{COMMIT} ,\,\@w{ABORT} \}}%
\@x{\@s{16.4} \.{\land} \E\, m \.{\in} msgs \.{:}}%
\@x{\@s{31.61} \.{\land} m \.{.} type \.{=}\@w{Ack}}%
\@x{\@s{31.61} \.{\land} m \.{.} dst \.{=} n}%
\@x{\@s{31.61} \.{\land} m \.{.} src \.{\in} children [ n ]}%
\@x{\@s{31.61} \.{\land} acks [ n ] [ m \.{.} src ] \.{=} {\FALSE}}%
\@x{\@s{31.61} \.{\land} acks \.{'} \.{=}}%
\@x{\@s{43.02} [ acks {\EXCEPT} {\bang} [ n ] [ m \.{.} src ] \.{=} {\TRUE} ]}%
\@x{\@s{31.61} \.{\land} {\UNCHANGED} {\langle} rmState ,\, children ,\,}%
\@x{\@s{43.02} intermediate\_children ,\,}%
\@x{\@s{43.02} msgs ,\, votes ,\, parent {\rangle}}%
\@pvspace{8.0pt}%
\@x{ ForgetCtx ( n ) \.{\defeq}}%
\@x{\@s{16.4} \.{\land} rmState [ n ] \.{\in} \{\@w{COMMIT} ,\,\@w{ABORT} \}}%
\@x{\@s{16.4} \.{\land} AllAcked ( n )}%
\@x{\@s{16.4} \.{\land} rmState \.{'} \.{=}}%
\@x{\@s{27.51} [ rmState {\EXCEPT} {\bang} [ n ] \.{=}\@w{TOMBSTONE} ]}%
\@x{\@s{16.4} \.{\land} {\UNCHANGED} {\langle} children ,\,}%
\@x{\@s{27.51} intermediate\_children ,\,}%
\@x{\@s{27.51} msgs ,\, votes ,\, acks ,\, parent {\rangle}}%
\@pvspace{8.0pt}%

\@x{ AddIntermediateParticipant ( n ,\, newChild ) \.{\defeq}}%
\@x{\@s{16.4} \.{\land} rmState [ n ] \.{\in}}%
\@x{\@s{27.51} \{\@w{RUNNING} ,\,\@w{PREPARE} ,\,\@w{COMMIT} ,\,\@w{ABORT} \}}%
\@x{\@s{16.4} \.{\land} newChild \.{\in} Node \.{\setminus} \{ n \}}%
\@x{\@s{16.4} \.{\land} newChild \.{\notin} children [ n ]}%
\@x{\@s{16.4} \.{\land} newChild \.{\notin} intermediate\_children [ n ]}%
\@x{\@s{16.4} \.{\land} intermediate\_children \.{'} \.{=}}%
\@x{\@s{27.51} [ intermediate\_children {\EXCEPT} {\bang} [ n ] \.{=}}%
\@x{\@s{38.92} @ \.{\cup} \{ newChild \} ]}%
\@x{\@s{16.4} \.{\land} {\UNCHANGED} {\langle} rmState ,\, children ,\,}%
\@x{\@s{27.51} msgs ,\, votes ,\, acks ,\, parent {\rangle}}%
\@pvspace{8.0pt}%

\@x{ Next \.{\defeq}}%
\@x{\@s{16.4} \.{\lor} RootStartToCommit}%
\@x{\@s{16.4} \.{\lor} \E\, n \.{\in} Node \.{:}}%
\@x{\@s{31.61} Handle2pcPrepareRequest ( n )}%
\@x{\@s{16.4} \.{\lor} \E\, n \.{\in} Node \.{:}}%
\@x{\@s{31.61} Handle2pcDuplicatePrepareRequest ( n )}%
\@x{\@s{16.4} \.{\lor} \E\, n \.{\in} Node \.{:}}%
\@x{\@s{31.61} HandleOrphan2pcPrepareRequest ( n )}%
\@x{\@s{16.4} \.{\lor} \E\, n \.{\in} Node \.{:}}%
\@x{\@s{31.61} Handle2pcPrepareResponse ( n )}%
\@x{\@s{16.4} \.{\lor} \E\, n \.{\in} Node \.{:}}%
\@x{\@s{31.61} Handle2pcCommitDecided ( n )}%
\@x{\@s{16.4} \.{\lor} \E\, n \.{\in} Node \.{:}}%
\@x{\@s{31.61} Handle2pcAbortDecided ( n )}%
\@x{\@s{16.4} \.{\lor} \E\, n \.{\in} Node \.{:}}%
\@x{\@s{31.61} Handle2pcCommitRequest ( n )}%
\@x{\@s{16.4} \.{\lor} \E\, n \.{\in} Node \.{:}}%
\@x{\@s{31.61} Handle2pcAbortRequest ( n )}%
\@x{\@s{16.4} \.{\lor} \E\, n \.{\in} Node \.{:}}%
\@x{\@s{31.61} HandleOrphan2pcCommitRequest ( n )}%
\@x{\@s{16.4} \.{\lor} \E\, n \.{\in} Node \.{:}}%
\@x{\@s{31.61} HandleOrphan2pcAbortRequest ( n )}%
\@x{\@s{16.4} \.{\lor} \E\, n \.{\in} Node \.{:}}%
\@x{\@s{31.61} InternalAbort ( n )}%
\@x{\@s{16.4} \.{\lor} \E\, n \.{\in} Node \.{:}}%
\@x{\@s{31.61} Handle2pcAckResponse ( n )}%
\@x{\@s{16.4} \.{\lor} \E\, n \.{\in} Node \.{:}}%
\@x{\@s{31.61} ForgetCtx ( n )}%
\@x{\@s{16.4} \.{\lor} \E\, n ,\, newChild \.{\in} Node \.{:}}%
\@x{\@s{31.61} AddIntermediateParticipant ( n ,\, newChild )}%
\@pvspace{8.0pt}%
\@x{ Spec \.{\defeq} Init \.{\land} {\Box} [ Next ]_{ Vars}}%
\@pvspace{8.0pt}%

\@x{ Consistency \.{\defeq}}%
\@x{\@s{16.4} \A\, n1 ,\, n2 \.{\in} Node \.{:}}%
\@x{\@s{27.51} {\lnot} ( rmState [ n1 ] \.{=}\@w{COMMIT}}%
\@x{\@s{38.92} \.{\land} rmState [ n2 ] \.{=}\@w{ABORT} )}%
\@pvspace{8.0pt}%
\@x{}\bottombar\@xx{}%